\documentclass[10pt]{article}

\usepackage[numbers]{natbib}
\usepackage[colorlinks,citecolor=blue,urlcolor=blue]{hyperref}
\usepackage{amsmath, amsthm, amssymb}
\usepackage{mathrsfs}
\usepackage{graphicx}
\usepackage{color}
\usepackage{xcolor}
\usepackage{dashrule}
\usepackage{booktabs}
\usepackage{wasysym}
\usepackage{verbatim}
\usepackage{multirow}
\usepackage{rotating}

\newcommand{\vV}[0]{\mathbf V}

\newcommand{\vX}[0]{\mathbf X}
\newcommand{\vY}[0]{\mathbf Y}
\newcommand{\vZ}[0]{\mathbf Z}

\newcommand{\vzero}[0]{\mathbf 0}
\newcommand{\vone}[0]{\mathbf 1}

\newcommand{\Cov}[0]{\text{Cov}}
\newcommand{\Var}[0]{\text{Var}}

\newcommand{\E}[0]{\mathbb{E}}

  \newcommand{\cW}{\mathcal{W}}

\newtheorem{thm}{Theorem}[section]
\newtheorem{lem}[thm]{Lemma}

\newtheorem{cor}[thm]{Corollary}

\newtheorem{rmk}{Remark}

\newcommand{\vbeta}[0]{\boldsymbol{\beta}}

\newcommand{\vdelta}[0]{\boldsymbol{\delta}}
\newcommand{\vmu}[0]{\boldsymbol{\mu}}
\newcommand{\veta}[0]{\boldsymbol{\eta}}

\newcommand{\vepsilon}[0]{\boldsymbol{\epsilon}}

\DeclareMathOperator*{\argmin}{argmin}

\graphicspath{{figures/}}

\title{Testing Linear Operator Constraints in Functional Response Regression with Incomplete Response Functions}
\date{}

\author{
Yeonjoo Park\\ University of Texas at San Antonio \\\\
Kyunghee Han\thanks{Corresponding author}\\ University of Illinois at Chicago \\\\
Douglas G. Simpson\thanks{Supported in part by National Institutes of Health Grant R01-CA226528-01A1.} \\ University of Illinois at Urbana-Champaign
}

\begin{document}

\maketitle

\begin{abstract}
Hypothesis testing procedures are developed to assess linear operator constraints in function-on-scalar regression when incomplete functional responses are observed. The approach enables statistical inferences about the shape and other aspects of the functional regression coefficients within a unified framework encompassing three incomplete sampling scenarios: (i) partially observed response functions as curve segments over random sub-intervals of the domain; (ii) discretely observed functional responses with additive measurement errors; and (iii) the composition of former two scenarios, where partially observed response segments are observed discretely with measurement error. The latter scenario has been little explored to date, although such structured data is increasingly common in applications. For statistical inference, deviations from the constraint space are measured via integrated $L^2$-distance between the model estimates from the constrained and unconstrained model spaces.  Large sample properties of the proposed test procedure are established, including the consistency, asymptotic distribution and local power of the test statistic. Finite sample power and level of the proposed test are investigated in a simulation study covering a variety of scenarios. The proposed methodologies are illustrated by applications to U.S. obesity prevalence data, analyzing the functional shape of its trends over time, and motion analysis in a study of automotive ergonomics.
\end{abstract}

\section{Introduction}
We develop a new scope of inferential procedures for testing the shape constraints on the regression coefficients in function-on-scalar regression models through the linear operator representation, where functional responses are incompletely observed. 
{\color{black} We assume that the functional response $Y_i(t)$ is available for $t \in \mathscr{I}_i$, where $\mathscr{I}_i \subset [0,1]$ is an individual-specific random subset independent of the stochastic mechanism that generates the complete functional response $Y_i$ for $i=1, \ldots, n$. 
We allow $\mathscr{I}_i$ for a union of sub-intervals, discrete subsets, or the composition of the two scenarios. 
The functional response fully available on the domain is a special case by simply letting $\mathscr{I}_i = [0,1]$. 
Recently \cite{kraus2015}, \cite{liebl2019partially}, \cite{delaigle2020} and \cite{kneip2020optimal} studied functional data analysis of partially observed curves, including principal component analysis, mean and covariance functions estimation, and optimal reconstruction of individual curves, but hypothesis testing problem has been less developed.
}

For statistical analysis, we assume that the unobservable complete functional response $Y$ is associated with vector covariates $\vX = (X_1,\ldots, X_p)^\top$ and $\vZ = (Z_1,\ldots, Z_q)^\top$ by the function-on-scalar regression model
\begin{equation} 
	\label{fullmodel}
	Y(t)= \mathbf{X}^{\top} \boldsymbol\beta(t) + \mathbf{Z}^{\top} \boldsymbol\alpha(t) + \epsilon(t) \quad (t \in [0,1]),
\end{equation}
 where $\boldsymbol\beta(t) = (\beta_1(t), \ldots, \beta_p(t))^\top$ and $\boldsymbol\alpha(t)$ $=$ $(\alpha_1(t),$ $\ldots,$ $\alpha_q(t))^\top$ are square-integrable vector coefficient functions, respectively, and  $\epsilon(t)$ is a mean-zero error process with the covariance function $\gamma(s,t) = \Cov\big(\epsilon(s), \epsilon(t)\big)$ independent of $(\vX, \vZ)$. 
 In the contexts of uncorrelated error processes or longitudinal data, the model \eqref{fullmodel} is also known as a varying coefficient regression model \citep{staniswalis1998nonparametric, hastie1993varying, malfait2003historical, eubank2004smoothing}. More recent developments have extended the theory and practice to functional and spatial varying coefficient models  \citep{zhu2012multivariate, zhu2014spatially, li2017functional, pietrosanu2021estimation, zhu2021network}. 

 {\color{black} We consider a class of testing composite null hypotheses on the functional regression model \eqref{fullmodel} of the form
\begin{equation}
    \label{shapespace}
    H_0: \boldsymbol{\mathcal{C}} \boldsymbol\beta = \mathbf{0},
\end{equation}
equivalently $H_0: \boldsymbol\beta \in \textrm{ker}(\boldsymbol{\mathcal{C}})$, 
where $\boldsymbol{\mathcal{C}}$ is a linear operator that maps vector functions to the function space of inferential interest, and $\textrm{ker}(\boldsymbol{\mathcal{C}})$ is the kernel space of $\boldsymbol{\mathcal{C}}$. 
In this study we focus on testing a dual formation of the null hypothesis \eqref{shapespace} expressed by
\begin{equation}
	\label{gnull}
	H_0: \boldsymbol\beta \in \textrm{span}\{ V(r) \},
\end{equation}
where $V(r) = \{v_{l} \in L^2:l=1, \ldots, r \}$ is an orthonormal basis that specifies the parametric family of $\textrm{ker}(\boldsymbol{\mathcal{C}})$.
It is worth mentioning that \eqref{gnull} is a generalization of the classical linear contrast hypothesis.
For example, let $\mathbf{C} \in \mathbb{R}^{d \times p}$ be of full rank $d$. 
The null hypothesis $H_0':\mathbf{C} \boldsymbol\beta(t) = \mathbf{0}$ studied in \cite{zhang2007statistical} identifies $\boldsymbol\beta(t) = \mathbf{u}_0(t) + \sum_{l=1}^d b_l \mathbf{u}_l$ for some vector function $\mathbf{u}_0(t) = (u_{0,1}(t), \ldots, u_{0,p}(t))^\top$ satisfying $\mathbf{C} \mathbf{u}_0(t) = \mathbf{0}$ and $\mathbf{b} = (b_1, \ldots, b_d)^\top \in \mathbb{R}^d$, where $U(d) = \{ \mathbf{u}_l \in \mathbb{R}^p:l=1, \ldots, d\}$ is the orthonormal basis of $\textrm{ker}(\mathbf{C})$ in $\mathbb{R}^p$.
{\color{black} Here we extend the theory from the finite dimensional constraints such as $H_0'$ to the potentially infinite dimensional linear operator constraints in \eqref{shapespace}. }

An important class of the null hypothesis \eqref{shapespace} includes testing the shape of regression functions. 
For example, \citep{gromenko2017evaluation} evaluated a physical mechanism for the conjectured linear trend in the Northern hemisphere cooling analysis. \cite{hejblum2015time} also tested linear or (piecewise) cubic time-course variations in gene expression experiments.
The functional trends can be evaluated by the coefficient function associated with the constant covariate $X=1$ and its shape constraint $H_0: \beta \in \textrm{span}\{V\}$, where $V$ is a set of $L^2$-functions or a basis that specifies the functional trend of interest. 
In this case, the space of shape-constrained regression functions is expressed by \eqref{shapespace}, where the kernel space of $\mathcal{C}$ is spanned by $V$.}
Previously \cite{ramsay2005} studied similar topics by testing individual probes 
in the form of $\langle c, \beta_j\rangle = 0$ for a fixed known $L^2$-function $c$ as a special case of \eqref{shapespace}. 
Later, \cite{james2006performing} proposed a residual-based permutation test for performing a hypothesis test on the shape of a mean function, although the large sample properties and the power behaviors of the proposed method were not investigated. 
Related work also includes \cite{yang2008hypothesis}, \cite{berkes2009detecting}, \cite{horvath2009two}, \cite{zhang2011statistical}, \cite{bugni2012specification}, \cite{hilgert2013minimax}, \cite{lei2014adaptive},  \cite{shang2015nonparametric}, \cite{staicu2015significance}, \cite{su2017hypothesis}, \cite{li2020inference},  \cite{garcia2021goodness}. 
Recently, \cite{cuesta2019goodness} and \cite{chen2020model} developed goodness-of-fit tests for functional models evaluated by empirical processes, and the significance of the family of models against general alternatives is tested by wild bootstrap resampling. 
But their applications to statistical inference are mainly aligned with validating functional linear models against a general class of non-structured functional models. 
Moreover, the extension of the existing methods to incomplete functional data has not been investigated. 

{\color{black} The main contributions of our study are as follows. We extend the the goodness-of-fit test to the general testing framework \eqref{gnull}, applicable to the model with incomplete functional response data. Our framework includes three scenarios as can often occur in practice; (i) the partially observed functional responses with random missing segments, where we have access to observations only for individual-specific sub-interval of the domain, but observation is not available on its complement, (ii) the functional responses observed with measurement errors on randomly spaced discrete evaluation points asynchronous across subjects, and (iii) a more challenging composite case, where individual curves are discretely observable over random sub-intervals of the domain. 
We especially investigate the theoretical property of the composition of the sub-interval censoring and discrete sampling of functional responses, where the proposed test procedure is applicable to a wide class of incomplete functional data. The asymptotic null distribution of the test statistic is also derived together with the consistency of the test with local alternatives $H_{1n}: \boldsymbol\beta = \boldsymbol\beta_0 + n^{-\tau/2} \boldsymbol\Delta$, where $\tau \in [0,1]$, for some $\boldsymbol\beta_0 = (\beta_{0,1}, \ldots, \beta_{0,p})^\top$ and $\boldsymbol\Delta = (\Delta_1, \ldots, \Delta_p)^\top$ satisfying $\boldsymbol{\mathcal{C}}\boldsymbol\beta_0 = \mathbf{0}$ and $\boldsymbol{\mathcal{C}} \boldsymbol\Delta \neq \mathbf{0}$, respectively.}


The methodology and basic theory of the proposed test procedures under incomplete functions responses are developed in Section \ref{sec:main}. In Section \ref{sec:sim}, we present numerical simulations, where the finite sample performance of the proposed test is evaluated in several scenarios. We also illustrate two applications from an obesity prevalence study and an automotive ergonomic experiment in Section \ref{sec:real-data}. Our concluding discussion is in Section \ref{sec:discussion}. Technical details, including the numerical implementation steps and theoretical proofs, are relegated to the Appendix.

\section{Main results} \label{sec:main}

\subsection{Partially sampled functional responses} \label{subsec:partial}

We first formulate the partially observed functional data as proposed in \cite{kraus2015}. 
Let $\delta_1, \ldots, \delta_n$ be a random sample of a stochastic process, defined on $[0,1]$, 
satisfying the following conditions.
\begin{itemize}
 \item [C1:] The latent stochastic processes, $(Y_i,\delta_i):=\{(Y_i(t),\delta_i(t)): t\in [0,1]\}$, for $i = 1,\ldots,n$, are independent and identically distributed on $(\Omega, \mathscr{F},\mathbb{P})$ and jointly $\mathscr{F}$-measurable.
  \item [C2:] $b(t) = E( \delta_i(t) )$ is bounded away from zero; i.e., $\inf_{t \in [0,1]}b(t) >0$
 \item [C3:] There are i.i.d. random variables $\boldsymbol{W}_i = (W_{i1}, \ldots, W_{iK}) \in {\cW}$, and there is a measurable function $h : [0,1] \times {\cW} \to \{0, 1\}$ such that $\delta_i(t) = h (t, \boldsymbol{W}_i)$.
   \item[C4:]  $Y_i$ and $\delta_i$ are independent for $i =1,\ldots,n$.
\end{itemize}
{\color{black} The partially observed functional responses are defined by $\{ Y_i(t): t \in \mathscr{I}_i, \, i=1, \ldots,n  \}$, where $\mathscr{I}_i = \{ t \in [0,1]: \delta_i(t) = 1 \}$ is the individual-specific random subset of $[0,1]$ for $i=1, \ldots, n$.} 
Various types of incomplete functional data structures satisfy conditions C1--C4, including dense functional snippets \citep{lin2020a}, fragmented functional data \citep{delaigle2020}, or {\color{black} functional data with single or multiple random missing intervals.  More examples can be found in \cite{Park2021}. Although C3 does not allow a sparse irregular sampling scheme, we consider the discretized noisy collection of partial data under the unified framework in Section \ref{subsec:composition}.}

\subsubsection{Estimation of functional regression coefficients and asymptotics}
{
Let $\vY^{\delta}(t) = (Y_1^\delta(t), \ldots, Y_n^\delta(t))^\top$ and $\boldsymbol\epsilon^\delta(t) = (\epsilon_1^\delta(t), \ldots, \epsilon_n^\delta(t))^\top$, where  $Y_i^\delta(t) = Y_i(t)$ and $\epsilon_i^\delta(t) = \epsilon_i(t)$ if $\delta_i(t)=1$, and $Y_i^\delta(t) = 0$ and $\epsilon_i^\delta(t) = 0$ otherwise for $t \in [0,1]$. That is, functional values over unobserved segments, $[0,1] \backslash \mathscr{I}_i $, are replaced by zeros. 
For an $n \times n$ diagonal matrix $\mathbb{W}(t) =  \mbox{diag}\{\delta_i(t) \}_{i=1}^n$, we write 
\begin{align}
    \mathbf{Y}^\delta(t) - \mathbb{W}(t)\mathbb{X} \boldsymbol\beta(t) = \mathbb{W}(t)\mathbb{Z} \boldsymbol\alpha(t) + \boldsymbol\epsilon^\delta(t)
    \label{fullmodel-re}
\end{align}
leads to $ \hat{\boldsymbol\alpha}^w(t; \boldsymbol\beta(t)) = (\mathbb{Z}^\top \mathbb{W}(t)\mathbb{Z})^{-1} \mathbb{Z}^\top \mathbb{W}(t)(\mathbf{Y}^\delta(t) -\mathbb{X} \boldsymbol\beta(t))$ the weighted least-squares estimator of $\boldsymbol\alpha(t)$, 
 where $\mathbb{X} = (\mathbf{X}_1, \cdots, \mathbf{X}_n)^\top$ and $\mathbb{Z} = (\mathbf{Z}_1, \cdots, \mathbf{Z}_n)^\top$ denote $(n \times p)$- and $(n \times q)$-design matrices of full rank, respectively. 
Substituting $\hat{\boldsymbol\alpha}^w(t; \boldsymbol\beta(t))$ for $\boldsymbol\alpha(t)$ in \eqref{fullmodel-re}, we obtain $(\mathbb{I} - \mathbb{P})\mathbf{Y}^\delta(t) = (\mathbb{I} - \mathbb{P})\mathbb{W}(t)\mathbb{X}\boldsymbol\beta(t) + \boldsymbol\epsilon^\delta(t)$, where $\mathbb{I} = \textrm{diag}(\mathbf{1}_n)$ and $\mathbb{P} = \mathbb{Z}(\mathbb{Z}^\top \mathbb{Z})^{-1} \mathbb{Z}^\top$ are the projection matrices that only depend on $\mathbf{1}_n$ and $\mathbb{Z}$. 
It follows that 
\begin{equation}
\label{beta_hat}
\hat{\boldsymbol\beta}^w(t)=   (\tilde{\mathbb{X}}^\top \mathbb{W}(t) \tilde{\mathbb{X}})^{-1} \tilde{\mathbb{X}}^\top \mathbb{W}(t) \mathbf{Y}^\delta(t) 
\end{equation}
is the weighted least-squares estimator of $\boldsymbol\beta(t)$, where $\tilde{\mathbb{X}} = (\mathbb{I} - \mathbb{P}) \mathbb{X}$ is the design matrix orthogonal to $\mathbb{Z}$.
This enables testing the hypothesis for $\boldsymbol\beta(t)$ while the nuisance regression coefficients related to $\mathbb{Z}$ are unspecified.
Indeed, $\hat {\boldsymbol\beta}^w(t)$ represents a pointwise least-square estimator calculated based on a subset of samples, where its response information is available at given \color{black} location \color{black} $t$. 
It also follows from $\hat{\boldsymbol\alpha}^w(t) = (\mathbb{Z}^\top \mathbb{Z})^{-1} \mathbb{Z}^\top \mathbb{W}(t)(\mathbf{Y}^\delta(t) - \mathbb{X} \hat{\boldsymbol\beta}^w(t))$ that 
\begin{equation}
    \hat{\boldsymbol\mu}(t) 
    = \tilde{\mathbb{X}} \hat{\boldsymbol\beta}^w(t) + \mathbb{Z}\hat{\boldsymbol\eta}^w(t) \label{mu-estimate}
\end{equation}
fits $\boldsymbol\mu(t) = E(\mathbf{Y}(t) \,|\, \mathbb{X}, \mathbb{Z})$ in a point-wise manner, where $\hat{\boldsymbol\eta}^w(t) = \hat{\boldsymbol\alpha}^w(t; \hat{\boldsymbol\beta}^w(t)) + (\mathbb{Z}^\top \mathbb{Z})^{-1} \mathbb{Z}^\top \mathbb{X} \hat{\boldsymbol\beta}^w(t)$. 
The expression \eqref{mu-estimate} will be used in the next subsection to define the model space.
}
\begin{thm}
\label{cor:regression}
Under $tr({\gamma}) < \infty$ and conditions C1--C4,
\begin{equation}
\begin{aligned}
	\sqrt{n} \big( \hat{\boldsymbol\beta}^w - {\boldsymbol\beta} \big)
	\stackrel{d}{\to} \textrm{GP}_p\big(\mathbf{0}_p, \vartheta  \Psi^{-1}  \big),
\end{aligned} 	\label{agp-pbeta0}
\end{equation}
where $\Psi= E(\mathrm{Var}(\vX | \vZ))$ and $\vartheta(s,t) = \gamma(s,t) \upsilon(s,t) \big/ b(s) b(t)$ with $\gamma(s,t) = \Cov(Y(s), Y(t))$,   $\upsilon(s,t) = E( \delta(s) \delta(t) )$, and $b(t) = E(\delta_i(t))$, for $s,t  \in [0,1]$. 
\end{thm}
Theorem \ref{cor:regression} implies that pointwise  ${\hat{\boldsymbol\beta}}^w(t)$ uniformly converges to $\vbeta(t)$ over $t \in [0,1]$ and further follows asymptotic Gaussian process with root-$n$ rates of convergence even under partial sampling structure. We also note that the condition on covariance function $tr({\gamma}) = \int_0^1 \gamma(t,t) \,\mathrm{d}t < \infty$ is commonly adopted in developing asymptotic theories on regression coefficient estimators under fully observed functional response. In practice, if we observe an undefined ${\boldsymbol\beta}^w(t)$ at a certain range of the domain under a
finite sample size, it can be estimated using interpolation or smoothing methods when the smoothness and continuity of $\vbeta(t)$ is assumed.

\subsubsection{The test statistic} \label{subsec:test-stat}

To test the appropriateness of the shape-constrained null hypothesis \eqref{shapespace} or equivalent \eqref{gnull}, we compare the model estimates from the unrestricted space $\mathcal{M} = \{ \vmu = \tilde{\mathbb{X}} \vbeta + \mathbb{Z} \veta: \beta_j \in L^2[0,1], \, j=1, \ldots, p \}$ and the reduced space $\mathcal{M}_0 = \{\vmu_0 = \tilde{\mathbb{X}} \vbeta_0 + \mathbb{Z} \veta:  \beta_{0,j} \in \textrm{span}\{ V(r) \}, \, j=1, \ldots, p  \}$. 
To this end, we construct a test statistic which is based on the $L^2$-distance between $\hat \vmu $ and $\hat \vmu_0$ defined by 
\begin{equation}
\begin{aligned}
    \hat \vmu 
    &= \argmin_{\boldsymbol{h} \in \mathcal{M}} \int_0^1 \| \mathbb{W}(t) \{\vY^\delta(t) - \boldsymbol{h}(t)\}  \|^2 \, \mathrm{d}t,\\
    \hat \vmu_0 
    &= \argmin_{\boldsymbol{h} \in \mathcal{M}_0} \int_0^1 \| \mathbb{W}(t)\{\vY^\delta(t) - \boldsymbol{h}(t)\} \|^2 \, \mathrm{d}t,
\end{aligned}
\end{equation}
where $\|\cdot \|$ denotes the standard $\ell^2$-norm in $\mathbb{R}^n$. The objective functions with the weight matrix $\mathbb{W}(t)$ imply the pointwise optimization under the partially sampled responses.
It can be verified that $\hat\vmu(t) = \tilde{\mathbb{X}} \hat{\boldsymbol\beta}^w(t) + \mathbb{Z} \hat{\boldsymbol\eta}^w(t)$ as in \eqref{mu-estimate}. 

Next, we define a linear operator $\mathcal{L}: L^2[0,1] \to \textrm{span}\{V(r)\}$ as 
\begin{equation}
	\mathcal{L} \beta = \sum_{l=1}^{r} \langle \beta, v_{l} \rangle v_{l} , 
\end{equation}
where $\langle f, g \rangle = \int_{0}^1 f(t) g(t) \, \mathrm{d}t $, and let $\boldsymbol{\mathcal{L}}$ denote the multivariate operator that applies $\mathcal{L}$ in an element-wise fashion. We then get $\hat{\vmu}_0(t) = \tilde{\mathbb{X}} \hat{\vbeta}_0^w(t) + \mathbb{Z} \hat{\boldsymbol\eta}^w(t)$, where $\hat\vbeta_0^w = \boldsymbol{\mathcal{L}}\hat\vbeta^w$. Even though we have partial response information for each observation, the uniformly consistent estimator $\hat{\boldsymbol\beta}^w$ provides the consistent model estimates $\hat\vmu(t)$ and $\hat\vmu_0(t)$ over $t \in [0,1]$. We next use them to propose a test statistic 
\begin{equation}
\begin{aligned}
T_n 
&= \int_0^1 \| \hat{\boldsymbol\mu}(t) - \hat{\boldsymbol\mu}_0(t) \|^2 \, \mathrm{d}t \\
&= \int_0^1 \big( \hat{\boldsymbol\beta}^w(t) - \hat{\boldsymbol\beta}_0^w(t) \big)^\top \big(\tilde{\mathbb{X}}^\top \tilde{\mathbb{X}} \big) \big(  \hat{\boldsymbol\beta}^w(t) - \hat{\boldsymbol\beta}_0^w(t) \big)\, \mathrm{d}t.
\end{aligned}\label{test-stat}
\end{equation}
Note that $T_n$ is the integrated squared distance between the model fits obtained under $\mathcal{M}$ and $\mathcal{M}_0$, respectively, and we reject the null hypothesis if $T_n$ is large.
Under the orthogonality between $\tilde{\mathbb{X}}$ and $\mathbb{Z}$, distance between $\hat \vmu$ and $\hat \vmu_0$ is translated to the weighted distance between two coefficient estimates $\hat \vbeta^w$ and $\hat \vbeta_0^w$.
While similar types of the $L^2$-norm based test-statistic have been employed in \cite{shen2004f}, \cite{zhang2007statistical}, and \cite{ zhang2011statistical} for conventional hypothesis testing, such as testing the nullity of functional coefficients, our study considers a more general scope of the null hypothesis, \color{black} using linear operator constraints, and the scope of response function sampling, allowing for  \color{black}  partially observed functional response data.

\subsubsection{Asymptotics and power considerations}
In this section, we derive the limit law of the proposed test statistic $T_n$ under the the null and local alternative hypotheses using the asymptotic Gaussianity of $\hat\vbeta^w$ shown in Theorem \ref{cor:regression}. Since a Gaussian process is closed under a linear operator and $\tilde{\mathbb{X}}^\top \tilde{\mathbb{X}}/n$ converges to $\Psi$ in probability, under the null hypothesis \eqref{gnull}, we can derive
\begin{equation}
\begin{aligned}
	\big(\tilde{\mathbb{X}}^\top \tilde{\mathbb{X}} \big)^{1/2} \big( \hat{\boldsymbol\beta}^w - \hat{\boldsymbol\beta}_0^w \big)
	&= \Psi^{1/2} \sqrt{n} \, \boldsymbol{\mathcal{C}} (\hat{\boldsymbol\beta}^w - \boldsymbol\beta_0 \big) + o_P(1)\\
	&\stackrel{d}{\to} \textrm{GP}_p\big(\mathbf{0}_p, \tilde\vartheta \mathbb{I}_p  \big),
\end{aligned} 	\label{agp-beta0}
\end{equation}
where $\boldsymbol{\mathcal{C}} = \boldsymbol{\mathcal{I}} - \boldsymbol{\mathcal{L}}$ for $\boldsymbol{\mathcal{I}}$ element-wisely operating the identity map $\mathcal{I}$ and
\begin{equation}
\begin{aligned}
	\tilde\vartheta(s,t)
	&=\vartheta(s,t)- \sum_{k=1}^r \bigg( \int_0^1 \vartheta(s,t) v_k(s) \, \mathrm{d}s \bigg) v_k(t) \\
 &\qquad\qquad- \sum_{l=1}^r \bigg( \int_0^1 \vartheta(s,t) v_l(t) \, \mathrm{d}t \bigg) v_l(s) \\
	&\qquad\qquad + \sum_{k=1}^r \sum_{l=1}^r \bigg( \int\int_{[0,1]^2} \vartheta(s,t) v_k(s) v_l(t) \, \mathrm{d}s \mathrm{d}t \bigg) v_k(s)v_l(t). 
\end{aligned} \label{gp-var}
\end{equation}


{\color{black}
We then consider sequences of local alternatives of the form

\begin{equation}
	\label{local-alternative}
	H_{1n}: \boldsymbol\beta = \boldsymbol\beta_0 + n^{-\tau/2} \boldsymbol\Delta,
\end{equation}
where $\tau \in [0, 1]$, and $\boldsymbol\Delta(t) = (\Delta_1(t), \ldots, \Delta_p(t))^\top$ represents a normalized functional deviation from the null hypothesis, independent of $n$. 
Then the asymptotic distribution of the test statistic is derived as the theorem below.

\begin{thm} \label{thm-alternative-dist}
	Suppose that $tr({\gamma}) < \infty$. And let $\{H_{1n}: n \geq 1 \}$ be a sequence of local alternatives with square-integrable functions $\Delta_j(t)$'s in \eqref{local-alternative}. Let $\tilde{\boldsymbol\Delta} = \Psi^{1/2} \, \boldsymbol{\mathcal{C}}\boldsymbol\Delta = (\tilde\Delta_1, \ldots, \tilde\Delta_p)^\top$ and define $\pi_m^2= \sum_{j=1}^p \|\langle \tilde\Delta_j, \phi_m \rangle \|^2 $, where $\phi_m$, $m=1,2, \ldots$, are eigenfunctions of $\tilde\vartheta(s,t)$. Then, the test statistic $T_n$ converges to $T_\Delta$ in probability, defined as
	\begin{equation}
		T_\Delta \stackrel{d}{=} \sum_{m=1}^\infty \lambda_m B_m, \label{thm-alternative-dist-eq}
	\end{equation}
	where $\lambda_m$ are decreasing-ordered eigenvalues of  $\tilde\vartheta(s,t)$, corresponding to eigenfunctions $\phi_m$, and $B_m \stackrel{i.i.d.}{\sim} \chi^2_{p}(\kappa_m^2)$ denotes the $p$ degrees of freedom non-central $\chi^2$-distribution with $\kappa_m^2 = \pi_m^2/\lambda_m$.
\end{thm}

Based on Theorem \ref{thm-alternative-dist}, we obtain the null distribution of the test statistic $T_n$ and asymptotic power derivations as follows.

\begin{cor} \label{thm-power} Assume the same conditions as in Theorem \ref{thm-alternative-dist}. 
\begin{enumerate}
    \item[(i)] Under the null hypothesis, i.e., $\boldsymbol{\mathcal{C}} \boldsymbol\Delta = \mathbf{0}$, Theorem \ref{thm-alternative-dist} implies that the null distribution of the test statistic $T_n$ converges to $T_0$ in distribution, where
	$T_0 = \sum_{m=1}^\infty \lambda_m A_m$
	with $\lambda_m$, decreasing-ordered eigenvalues of $\tilde\vartheta(s,t)$, and $A_m \stackrel{i.i.d.}{\sim} \chi^2_p$.

	
\item[(ii)] Suppose that $\boldsymbol{\mathcal{C}} \boldsymbol\Delta \neq \mathbf{0}$, that is, the local alternative, and $\sum_{m=1}^\infty \pi_m^2 = \infty$ or $\tau \in [0,1)$. Then, Theorem \ref{thm-alternative-dist} yields the asymptotic power of the test as; $\lim_{n \to \infty} P(T_n \geq t_\alpha | H_{1n}) = 1$, where $t_\alpha$ is the upper-$\alpha$ quantile of the null distribution $T_0$ in the case (i). 
\end{enumerate}
\end{cor}


}

As we can see in the proof of the Corollary \ref{thm-power} in the Appendix, the asymptotic power of the test goes 1 under $H_{1n}$ of \eqref{local-alternative} with any $\tau \in [0, 1)$ and non-zero $\Delta$, which is desirable. When considering $\tau = 1$, where the local alternative tends to the null with root-$n$ rate, the non-trivial asymptotic power goes to 1 when $\sum_{m=1}^\infty \pi_m^2 = \infty$. In Section \ref{sec:sim} of simulation studies, we consider different magnitudes of null-deviated signals $\pi_m^2$ under $\tau=1$ to investigate the power in the practical setting.

\subsection{Discrete observations with measurement errors} \label{sec:kernel}

In this section, we extend the proposed test to the case where functional responses are observed with measurement errors over finite discrete points in their domains, and possibly sampled asynchronously across subjects. Being different from the partially observed sampling scheme with continuum measurement over the subset of $[0,1]$, we consider the case that functional measurements are collected on a discrete subset of the functional domain with additive measurement errors. Specifically, let $\{ (\mathbf{Y}_i^\ast, \mathbf{T}_i, \mathbf{X}_i, \mathbf{Z}_i): i=1, \ldots, n\}$ be a random sample of $(\mathbf{Y}^\ast, \mathbf{T}, \mathbf{X}, \mathbf{Z})$, where $\mathbf{Y}_i^\ast = (Y_{i,1}^\ast, \ldots, Y_{i,N_i}^\ast)^\top$ is the finite observations of the $i$-th subject associated with evaluation points $\mathbf{T}_i = (T_{i,1}, \ldots, T_{i,N_i})^\top$ as 
\begin{equation}
\begin{aligned}
	Y_{i,m}^\ast
	= \mathbf{X}_i^{\top} \boldsymbol\beta(T_{i,m}) + \mathbf{Z}_i^{\top} \boldsymbol\alpha(T_{i,m}) + \epsilon_i(T_{i,m})  + \varepsilon_{i,m}.
\end{aligned} \label{model-longitudinal}
\end{equation}
The sampling design differs from the ones considered in the previous subsections as functional outcomes are prone to measurement errors, denoted by $\varepsilon_{i,m}$, and finite observations are only available. We note that $\varepsilon_{i,m} = 0$ is a special case that follows the same model \eqref{fullmodel}. 
For statistical analysis, we assume that $\varepsilon_{i,1}, \ldots, \varepsilon_{i,N_i}$ are i.i.d. as $\varepsilon$ such that $E(\varepsilon | \mathbf{X}, \mathbf{Z}) = 0$ and $E |\varepsilon|^k < \infty$ for some $k > 2$. The finite evaluation points $T_{i,1}, \ldots, T_{i,N_i}$ are randomly generated by a probability density function $\lambda(t)$ bounded away from zero and infinity whose derivative also exists and is bounded. We also assume that $N_1, \ldots, N_n$ are i.i.d. as an independent random integer $N \geq 1$ that asymptotically increases as the sample size $n$ becomes large. This sampling framework is similar to those considered by
\cite{yao2005functional, zhang2013time, petersen2016functional}, and \cite{han2020additive}. 
We refer to the theorem and remark below for technical details.

However, it is infeasible to apply the same procedure demonstrated in Section \ref{subsec:partial} because functional responses are only available at discrete evaluation points. Unlike the partially sample functional responses, we lose functional continuum in outcome variables which bases point-wise estimates \eqref{beta_hat} to calculate the test statistic \eqref{test-stat}. More importantly, the magnitude of false signals is not ignorable in the presence of measurement errors. This means that, even though infinitely many evaluation points are available, the coefficient function estimates will be biased as we may not achieve consistency. As a result, Corollary \ref{thm-power} may not serve as a reference distribution for testing \eqref{gnull}.

To tackle the bottleneck, we employ kernel smoothing to recover the unobserved functional responses, where the false signals produced by measurement errors are mitigated, and substitute the estimated curves for the true functional responses to perform the test demonstrated in Section \ref{subsec:partial}. Formally, as a two-step procedure, we first kernel smooth discrete observations for each subject as the Nadaraya-Watson kernel estimator of $E(Y_i(T) \,|\, T = t)$
\begin{equation} \label{kernsmooth}
	\tilde{Y}_i^\ast(t) = \frac{\sum_{m=1}^{N_i} K_h(T_{i,m} - t) Y_{i,m}^\ast}{\sum_{m'=1}^{N_i} K_h(T_{i,m'} - t) } \quad (t \in [0,1])
\end{equation}
for each $i=1, \ldots, n$. $\tilde{Y}_i^\ast(t)$ is , where $h > 0$ is a bandwidth and $K_h(t) = K(t/h)/h$ is the scaled kernel of a symmetric density function $K$. 
Then, we define a kernel-smoothed test statistic as
\begin{equation}
	T_n^\ast
	= \int_0^1 \big( \tilde{\boldsymbol\beta}^\ast(t) - \tilde{\boldsymbol\beta}^\ast_0(t) \big)^\top \big(\tilde{\mathbb{X}}^\top \tilde{\mathbb{X}} \big) \big( \tilde{\boldsymbol\beta}^\ast(t) - \tilde{\boldsymbol\beta}^\ast_0(t) \big)\, \mathrm{d}t,
    \label{kernel-test-stat}
\end{equation}
where $\tilde{\boldsymbol{\beta}}^\ast(t) = (\tilde{\mathbb{X}}^{\top} \tilde{\mathbb{X}})^{-1} \tilde{\mathbb{X}}^{\top} \tilde{\mathbf{Y}}^\ast(t)$ and $\tilde{\boldsymbol{\beta}}_0^\ast = \mathcal{L} \tilde{\boldsymbol{\beta}}^\ast$. Finally, we reject the null hypothesis \eqref{gnull} if $T_n^{\ast} > t_{\alpha}$, where $t_{\alpha}$ is the level-$\alpha$ critical value for $T_0$ in Corollary \ref{thm-power}.

In the pre-smoothing approach, it is critical to recover individual curves with a uniform rate of convergence on the entire domain $[0,1]$ since the kernel-smoothed test statistic $T_n^\ast$ is defined as a weighted $L^2$ norm of $\, \boldsymbol{\mathcal{C}}\tilde{\boldsymbol\beta}^\ast$, while $\tilde{\boldsymbol\beta}^\ast(t)$ is given by a point-wise estimate.  
{\color{black} The local constant smoothing, also known as Nadaraya-Watson type estimation, is easy to implement, but it is less preferred when the reconstruction of individual curves is of main interest in functional data analysis because the asymptotic bias near the boundary of the domain may vary with individual smoothing.
However, Theorem \ref{thm-kernel-test-consistent} below shows that we can still attain the consistency of the test procedure with local constant smoothing. }

\begin{thm} \label{thm-kernel-test-consistent}
	Assume that $E \|Y\|_\infty^k < \infty$ for some $k > 2$ and $ \max_{1 \leq i \leq n} \| Y' \|_\infty$ is bounded in probability. If $h \asymp n^{-\theta/5}$ and $P(N < n^\theta) = o(n^{-1})$ for some $\theta > 5/3$, then $T_n^\ast - T_n = o_P(1)$, where we define $T_n$ in \eqref{test-stat} with $\delta_i = 1$ for all $i=1, \ldots, n$. 
\end{thm}


\begin{rmk}
    For each $i$-th subject, the optimal rate of univariate bandwidth for kernel estimation is typically given by $h \asymp N_i^{-1/5}$ \citep{hall1983large, hall1991local}. Since $N_i \geq n^\theta$ for all $1 \leq i \leq n$ with probability tending to $1$ (Lemma \ref{thm-kernel-unif-conv} in the Appendix), the use of a common rate $h \asymp n^{-\theta/5}$ in Theorem \ref{thm-kernel-test-consistent} allows us to employ the existing bandwidth selectors \citep{park1990comparison, jones1996brief}.
\end{rmk}

{\color{black}
However, we note that the classical pre-smoothing approach such as \cite{ramsay2005} requires densely observed functional responses over the entire domain for all subjects. In practice, this requirement is implausible when the observations are relatively sparse. In the following subsection, we introduce a new scope of partially observed functional data to ease the limitation. 
}

\subsection{Composition of partial filtering and discrete sampling} \label{subsec:composition}

In Section \ref{subsec:partial}, partially observed data were assumed to be evaluated over continuous subsets of the functional domain. If such data are observed discretely rather than continuously, then the observation framework reduces to that of discretely observed functional response data, and the smoothing approach of Section \ref{sec:kernel} may be applied. In this case, we assume that the complete observations for responses are given by $\{Y_i^\delta: i=1, \ldots, n\}$. For random evaluation points $\mathbf{T}_i = (T_{i,1}, \ldots, T_{i, N_i})^\top$ and the indicator process $\delta_i$, we define a random subset $\mathscr{I}_i^\ast = \{ j : \delta_i(T_{i,j}) = 1, \, j=1, \ldots, N_i \}$. We assume that $\mathbf{T}_i$ and $\delta_i$ are independent. 

The corresponding discrete functional observations are given by $\{Y_{i,m}^\ast, T_{i,m}^\ast : m \in \mathscr{I}_i^\ast\}$, where $Y_{i,m}^\ast = Y_i^\delta(T_{i,m}^\ast) + \varepsilon_{i,m}$ and $T_{i,m}^\ast = T_{i,j}$ for some $j=j_m \in \mathscr{I}_i^\ast$, which can be viewed as the discrete sampling of $Y_i$ composed with the partial filtering process $\delta_i$. Then, we define
\begin{equation} \label{kernsmooth-composition}
	\tilde{Y}_i^{\ast\ast}(t) = \frac{\sum_{m=1}^{N_i} K_h(T_{i,m}^\ast - t) Y_{i,m}^\ast}{\sum_{m'=1}^{N_i} K_h(T_{i,m'}^\ast - t) } \quad (t \in \mathscr{I}_i)
\end{equation}
for each $i=1, \ldots, n$. Also, we define a kernel-smoothed test statistic as
\begin{equation}
	T_n^{\ast\ast}
	= \int_0^1 \big( \tilde{\boldsymbol\beta}^{\ast\ast}(t) - \tilde{\boldsymbol\beta}^{\ast\ast}_0(t) \big)^\top \big(\tilde{\mathbb{X}}^\top \tilde{\mathbb{X}} \big) \big( \tilde{\boldsymbol\beta}^{\ast\ast}(t) - \tilde{\boldsymbol\beta}^{\ast\ast}_0(t) \big)\, \mathrm{d}t,
    \label{kernel-test-stat-composition}
\end{equation}
where $\tilde{\boldsymbol{\beta}}^{\ast\ast}(t) = (\tilde{\mathbb{X}}^{\top} \mathbb{W}(t)\tilde{\mathbb{X}})^{-1} \tilde{\mathbb{X}}^{\top} \mathbb{W}(t)\tilde{\mathbf{Y}}^{\ast\ast}(t)$ and $\tilde{\boldsymbol{\beta}}_0^{\ast\ast} = \mathcal{L} \tilde{\boldsymbol{\beta}}^{\ast\ast}$. 

To investigate the theoretical property of the proposed method, we assume that
\begin{align}
    E\Bigg\vert \frac{1}{\int_0^1 \delta_i(v) \, \mathrm{d}v} \Bigg\vert^p < \infty.
    \label{delta-condition1}
\end{align}
for some $p > 2$. Also, 
suppose that there exists an absolute constant $C > 0$ satisfying
\begin{align}
\begin{split}
   P(\delta_i(s) \neq \delta_i(t)) \leq C|s - t|^p
\end{split} \label{delta-condition3}
\end{align}
{\color{black} The reciprocal moment condition \eqref{delta-condition1} implies that that the length of the random sub-interval $\mathscr{I}_i  = \int_0^1 \delta_i(v) \, \mathrm{d}v$ is positive (a.s.). Hence, together with \eqref{delta-condition3}, the composition sampling has discrete observations densely available on each sub-interval, but not necessarily over the entire domain. We also refer to Remark \ref{rmk-condition} below for the equivalent expression of \eqref{delta-condition3}.}

\begin{thm} \label{thm:section3}
    Assume the same conditions as Theorem \ref{cor:regression} and Theorem \ref{thm-kernel-test-consistent}. If \eqref{delta-condition1} and \eqref{delta-condition3} hold for some $p > 2$, then $T_n^{\ast\ast} - T_n = o_P(1)$.
\end{thm}

\begin{proof}
We note that
\begin{align}
\begin{split}
    P(T_{i,m}^\ast \in A) 
    &= E\big[P(T_{i,j_m} \in A \,|\, \delta_i(T_{i,j_m}) = 1)\big] \\
    &= E\bigg[\frac{P(T_{i,j_m} \in A, \delta_i(T_{i,j_m}) = 1 \,|\, \delta_i)}{P(\delta_i(T_{i,j_m}) = 1 \,|\, \delta_i)}\bigg]\\
    &= E\bigg[\frac{\int_ A \delta_i(u) \lambda(u) \, \mathrm{d}u}{\int_0^1 \delta_i(v) \lambda(v) \, \mathrm{d}v }\bigg].
\end{split} \nonumber
\end{align}
The above observation implies that, even if discrete observations are sampled from the random segments of functional responses, the proposed method works for this case if we impose additional conditions on the filtering process $\delta_i$ so that the density of $T_{i,m}^\ast$ given by 
\[
    \lambda^\ast(t) = E\bigg[\frac{\delta_i(t) \lambda(t)}{\int_0^1 \delta_i(v) \lambda(v) \, \mathrm{d}v }\bigg] 
    \quad (t \in [0,1])
\]
satisfies the key design condition of Theorem \ref{thm:section3}.

For the boundedness of $\lambda^\ast$, we note that conditions (C1)-(C4) and the assumptions on $\lambda$ imply the uniform lower bound,
\[
\lambda^\ast(t) \ge E\bigg[\frac{\delta_i(t)\lambda(t)}{\Vert \lambda \Vert_\infty}\bigg] = \frac{b(t)\lambda(t)}{\Vert \lambda \Vert_\infty}
\ge \frac{b_0 \lambda_0}{\Vert \lambda\Vert_\infty} > 0,
\]
where $b_0 = \inf_t b(t)$ and $\lambda_0 = \inf_t \lambda(t)$.
Also, \eqref{delta-condition1} gives the uniform upper bound,
\begin{align}
\lambda^\ast(t) \le \frac{\Vert \lambda \Vert_\infty}{\lambda_0}
E\bigg[\frac{1}{\int_0^1 \delta_i(v) dv}\bigg]. \label{lambda-star-bounded}
\end{align}

For the smoothness of $\lambda^\ast$, we note that
\begin{align}
        \begin{split}
            \frac{\lambda^\ast(s) - \lambda^\ast(t)}{s-t}
            &= \frac{1}{s-t}E\bigg[\frac{\delta_i(s) \lambda(s) - \delta_i(t) \lambda(t)}{\int_0^1 \delta_i(v) \lambda(v) \, \mathrm{d}v }\bigg]\\
            &= \sum_{j=1}^3 E\bigg[\frac{A_{ij}(s,t)}{\int_0^1 \delta_i(v) \lambda(v) \, \mathrm{d}v }\bigg], 
        \end{split} \label{lambda-star-diff}
        \end{align}
where $A_{i1}(s,t) = \frac{\lambda(s) - \lambda(t)}{s-t} \, \mathbb{I}(s, t \in \mathscr{I}_i)$, $A_{i2}(s,t) = \frac{\lambda(s)}{s-t} \, \mathbb{I}(s\in \mathscr{I}_i, \, t \not\in \mathscr{I}_i)$, and $A_{i3}(s,t) = - \frac{\lambda(t)}{s-t} \, \mathbb{I}(s\not\in \mathscr{I}_i, \, t \in \mathscr{I}_i)$. Obviously, $|A_{i1}(s,t)|$ is bounded (a.s.) since $\lambda$ has a bounded derivative. The moment condition \eqref{delta-condition1} and the dominated convergence theorem give
\begin{align}
    \lim_{s \to t} E\bigg[ \frac{A_{1j}(s,t)}{\int_0^1 \delta_i(v) \lambda(v) \, \mathrm{d}v } \bigg] =  E\bigg[\frac{\delta_i(t)\lambda'(t) }{\int_0^1 \delta_i(v) \lambda(v) \, \mathrm{d}v } \bigg]. \label{lambda-star-derivative}
\end{align}

To analyze $A_{i2}(s,t)$, using H\"older's inequality with $p^{-1}+q^{-1}=1$ for $p, q >1$, we have
\[
E\Bigg\vert \frac{A_{i2}(s,t)}{\int_0^1 \delta_i(v)\lambda(v)dv}\Bigg\vert \le
\frac{\Vert\lambda\Vert_\infty} {\lambda_0} \left\{ E\bigg\vert \frac{1}{\int_0^1 \delta_i(v) dv} \bigg\vert^p\right\}^{1/p}
\left\{E\bigg\vert \frac{\mathbb{I}(s\in \mathscr{I}_i, \, t \not\in \mathscr{I}_i)}{s-t}\bigg\vert^q\right\}^{1/q}.
\]
Doing the same with $A_{i3}$, we claim that 
\begin{align}
\limsup_{s\to t} \frac{P(\delta_i(s) \neq \delta_i(t))}{\vert s-t \vert^q}
= 0 \label{delta-condition2}
\end{align}
Indeed, it follows from \eqref{delta-condition3} that
\begin{align}
\begin{split}
    \frac{P(\delta_i(s) \neq \delta_i(t))}{\vert s-t \vert^q}
    &\leq 2C |s-t|^{p-q},
\end{split} \nonumber 
\end{align}
provided that $p > 2 > \frac{p}{p-1}= q$. 
Therefore, combining \eqref{lambda-star-diff}, \eqref{lambda-star-derivative}, and \eqref{delta-condition2}, we conclude that the derivative of $\lambda^\ast$ is given by
\[
    (\lambda^\ast)'(t) = E\bigg[\frac{\delta_i(t)\lambda'(t) }{\int_0^1 \delta_i(v) \lambda(v) \, \mathrm{d}v } \bigg].
\]
The boundedness of $(\lambda^\ast)'$ can also be shown similarly as \eqref{lambda-star-bounded}.
\end{proof}

\begin{rmk} \label{rmk-condition}
    The condition \eqref{delta-condition3} can also equivalently understood as
    {\color{black}
    \begin{align}
        \begin{split}
            \vert \Gamma(s,t) - b(s)(1-b(t))\vert \leq C|s - t|^p
        \end{split} \label{delta-condition3-equiv}
    \end{align}
    as well as $\vert \Gamma(s,t) - b(t)(1-b(s))\vert \leq C|s - t|^p$, 
    where $\Gamma(s,t) = \mathrm{Cov}\big(\delta_i(s), \delta_i(t)\big)$ and $\Gamma(t,t)=b(t)(1-b(t))$}. To see this, we note that
    \begin{align}
\begin{split}
    \Gamma(s,t)
     &= \mathrm{Cov}\big(\delta_i(s), \delta_i(t)\big)\\
    &= E\big[ \delta_i(s) \delta_i(t) \big] - E\big[ \delta_i(s) \big]E\big[ \delta_i(t) \big]\\
    &= P\big(\delta_i(s) = 1,\, \delta_i(t) = 1\big) - b(s)b(t).
\end{split} \nonumber
\end{align}
Similarly, we have 
 \begin{align}
 \begin{split}
     \Gamma(s,t)
     &= \mathrm{Cov}\big(1-\delta_i(s), 1-\delta_i(t)\big)\\
     &= P\big(\delta_i(s) = 0,\, \delta_i(t) = 0\big) - (1-b(s))(1-b(t)).
 \end{split} \nonumber
 \end{align} 
It follows that
 \begin{align}
\begin{split}
     P\big( \delta_i(s) \neq \delta_i(t) \big)
     &= 1 - P\big(\delta_i(s)=0,\, \delta_i(t) = 0\big) - P\big(\delta_i(s)=1, \delta_i(t)=1\big)\\
     &= \big\{b(s)(1-b(t)) - \Gamma(s,t)\big\} + \big\{b(t)(1-b(s)) - \Gamma(s,t) \big\}.
\end{split} \nonumber 
\end{align}
Indeed, 
    \begin{align}
        \begin{split}
            b(s)(1-b(t)) - \Gamma(s,t)
            &= b(s)(1-b(t)) - E\big[ \delta(s)\delta(t) \big] + b(s)b(t) \\
            &= P\big( \delta_i(s) = 1 \big) - P\big( \delta_i(s) = 1,\, \delta_i(t) = 1 \big)\\
            &= P\big( \delta_i(s) = 1,\, \delta(t)_i = 0 \big).
        \end{split} \nonumber
    \end{align}
    Similarly, we have 
    \begin{align}
    \begin{split}
        b(t)(1-b(s)) - \Gamma(s,t) 
            &= P\big( \delta_i(t) = 1 \big) - P\big( \delta_i(s) = 1,\, \delta_i(t) = 1 \big)\\
            &= P\big( \delta_i(s) = 0,\, \delta_i(t) = 1 \big).
    \end{split} \nonumber
    \end{align}
    Therefore, \eqref{delta-condition3} and \eqref{delta-condition3-equiv} are equivalent because
    \begin{align}
    \begin{split}
        P(\delta_i(s) \neq \delta_i(t)) 
        &= b(s)(1-b(t))+b(t)(1-b(s))-2\Gamma(s,t). 
    \end{split}
    \end{align}
\end{rmk}

\begin{rmk}\label{rmk:delta}
    We provide one simple example of $\delta$ that satisfies the conditions \eqref{delta-condition1} and \eqref{delta-condition3}. 
    Suppose that $U_{(1)} < \cdots < U_{(2p+k+2)}$ be order statistics of a $\mathrm{Uniform}(0,1)$ random sample of size $(2p+k+2)$ for some $k \geq p = 3$. Let $\delta(t) = \mathbb{I}(U_{(p+1)} \leq t \leq U_{(p+k+2)})$. Since $S = U_{(p+k+2)} - U_{(p+1)}$ has a $\mathrm{Beta}(k+1, 2p+2)$ distribution, the condition \eqref{delta-condition1} holds, i.e, 
\[
    E\Bigg\vert \frac{1}{\int_0^1 \delta(v) dv} \Bigg\vert^p = E|1/S^p| = \frac{(k-p)! \, (k+2p+2)!}{k! \, (k+p+2)!
    } < \infty.
\]
To verify the condition \eqref{delta-condition3}, let $s < t$ without loss of generality. We note that
    \begin{align}
        \begin{split}
            P\big( \delta(s) = 1,\, \delta(t) = 0 \big) 
            &= P(U_{(p+1)} \leq s \leq  U_{p+k+1} < t)\\
            &= C_{p,k} \int_s^t \int_0^s u^3 (v-u)^k (1-v)^3 \, \mathrm{d}u \mathrm{d}v,
        \end{split} \nonumber
    \end{align}
    where $C_{p,k} = \frac{(2p+k+2)!}{p! \, k! \, p!}$. 
    For $g(s,t) = \int_s^t \int_0^s u^3 (v-u)^k (1-v)^3 \, \mathrm{d}u \mathrm{d}v$ satisfying $g(s,s) = 0$, the Leibniz rule and integration by parts give
    \begin{align}
        \begin{split}
            g^{(0,1)}(s,t) 
            = \frac{\partial}{\partial t} g(s,t)
            &= (1-t)^3 \int_0^s u^3 (t-u)^k \, \mathrm{d}u\\
            &= (1-t)^3 \sum_{\ell=0}^3 c_\ell s^{3-\ell} (t-s)^{k+ 1 + \ell}
        \end{split}
    \end{align}
    for some non-zero constants $c_0, \ldots, c_3$. Therefore, it follows from the mean value theorem that
    \begin{align}
        \begin{split}
            P\big( \delta(s) = 1,\, \delta(t) = 0 \big)
            &\leq C_{p,k}\big| g(s,t) - g(s,s) \big|\\
            &\leq C_{p,k} |s-t| \sup_{u \in [s,t]}\big| g^{(0,1)}(s,u) \big| \\
            &\leq C_{p,k}^\ast |s-t|^{k+2},
        \end{split}
    \end{align}
    where $C_{p,k}^\ast = C_{p,k} \max_{\ell} c_\ell$. The case for $P\big( \delta(s) = 0,\, \delta(t) = 1 \big) $ can also be verified similarly, and we get the condition \eqref{delta-condition3}.
\end{rmk}


\section{Simulation studies}\label{sec:sim}
In this section, we study the finite sample performance of the proposed testing procedure in terms of size control and powers under various settings. The performances under incomplete functional response models are compared to the benchmark performance, where functional responses are fully observed without measurement errors.

\subsection{Simulation setting}
We first generate the fully observed response $Y_i$ from the model
\begin{equation}\label{simY}
Y_i(t) = \mathbf{X}_i^\top \boldsymbol\beta(t) + \mathbf{Z}_i^\top \boldsymbol\alpha(t) + \epsilon_i(t), \quad (t \in [0,1])
\end{equation}
for $i=1,\ldots, n$, where covariates $\mathbf{X}_i = (1_{\{U_{i1} > 0\}}, \Phi(U_{i2}), U_{i3})^\top$ and $\mathbf{Z}_i= (1, U_{i4})^\top$ are from $\mathbf{U}_i \stackrel{i.i.d.}{\sim} N_4(\mathbf{0}, \Sigma)$ with $\Sigma = [\sigma_{ij}]_{1\leq i,j \leq 4}$ for $\sigma_{ij}= 0.5^{|i - j|}$, and $\Phi$ denoting the cdf of $N(0,1)$; functional coefficients $\boldsymbol\alpha(t) $ $=$ $\{\alpha_1(t),$ $\alpha_2(t) \}^\top$ associated with $\mathbf{Z}_i$ are generated by $\alpha_k(t) = \sum_{l=4}^{5} ( k+l)^{-1/2} (-1)^{l}v_l(t) \big/ \{\sum_{l=4}^{5} (k+l)^{-1}\}^{1/2}$ for $k=1, 2$, where $V(5) = \{ v_l(t);~t \in [0,1]\}_{l=1}^5$ is a set of orthonormal polynomial base derived from polynomials $P(5) = \{ t^{l-1}; ~t\in [0,1]\}_{l=1}^5$, that is, $\alpha_k \in \textrm{span}\{V(5)\}$ satisfying $\| \alpha_k \|_2 = 1$; random error is independently and identically generated from $\epsilon_i(t) = \sum_{m=1}^{100} e_m \phi_m(t)$, where $\phi_m(t) = \sqrt{2}\sin(2m\pi t)$ and $e_m \stackrel{i.i.d.}{\sim} N(0, 4 m^{-4})$, for $m=1, \ldots, 100$. Functional trajectories are generated at a regular grid of 100 points in $[0, 1]$ and the sample size $n$ is chosen to be 100 and 200.

Let $\vbeta_0(t)=\{\beta_{0,1}(t),\beta_{0,2}(t),\beta_{0,3}(t) \}^\top$, where $\beta_{0,j}(t) = \{v_1(t) + v_{j+1}(t)\}/\sqrt{2}$, for $j=1,2,3$, implying that $\beta_{0,j} \in \textrm{span}\{V(4)\}$ and $\|\beta_{0,j}\|_2=1$. We then consider two scenarios A and B on $\vbeta(t)$. In scenario A, we set $\vbeta(t) = \vbeta_0(t) + n^{-\tau/2} \{d_A  \vdelta_A(t)\}$, where $d_A>0$ and  $\vdelta_A(t)=\{\delta_{A,1}(t),\delta_{A,2}(t), \delta_{A,3}(t)\}^T$ with $\delta_{A,j}(t) = \sum_{m=1}^{100} (j+m)^{-1/2} \phi_m(t) \big / \{\sum_{m=1}^{100} (j+m)^{-1}\}^{1/2}$. And we consider a hypothesis testing for the null hypothesis
\begin{equation}\label{hyp:sim}
H_0: \beta_j \in \textrm{span}\{V(4)\}, \quad \forall j=1,2,3.
\end{equation}
It aims to find statistical evidence on whether $\beta_j(t)$ coefficients can be expressed exclusively by polynomials up to order three. We investigate the empirical size and power of the proposed method under different magnitudes of the null-deviated signals by setting $d_A=0,1,3,5,7,9$. For each $d_A$, we further set $\tau=1, 0.8, 0.67$, corresponding to the rates of the local alternative approaching to the null as $n^{1/2}$, $n^{1/2.5}$, $n^{1/3}$, respectively,  to examine the performance under different rates that the null-deviated model tends to the null model.
In scenario B, we consider a test for the same hypothesis of  \eqref{hyp:sim} under $\vbeta(t) = \vbeta_0(t) + n^{-\tau/2} \{d_B \vdelta_B(t)\}$, where $\vdelta_B(t) = \{\delta_{B,1}(t),\delta_{B,2}(t), \delta_{B,3}(t)\}^T$ with $\delta_{B,j}(t) =v_5(t)$. we set $d_B=0, 0.3, 0.6, 0.9, 1.2, 1.5$, and $\tau=1, 0.8, 0.67$. Figure \ref{fig:simalt} illustrates deviations of $\vbeta(t)$ from $\vbeta_0(t)$ under two scenarios for $d_A=3$ and $d_B=0.6$, respectively, when $\tau=1$ and $n=100$.

\begin{figure}[t!]
  \centering
  \includegraphics[width=4.8in]{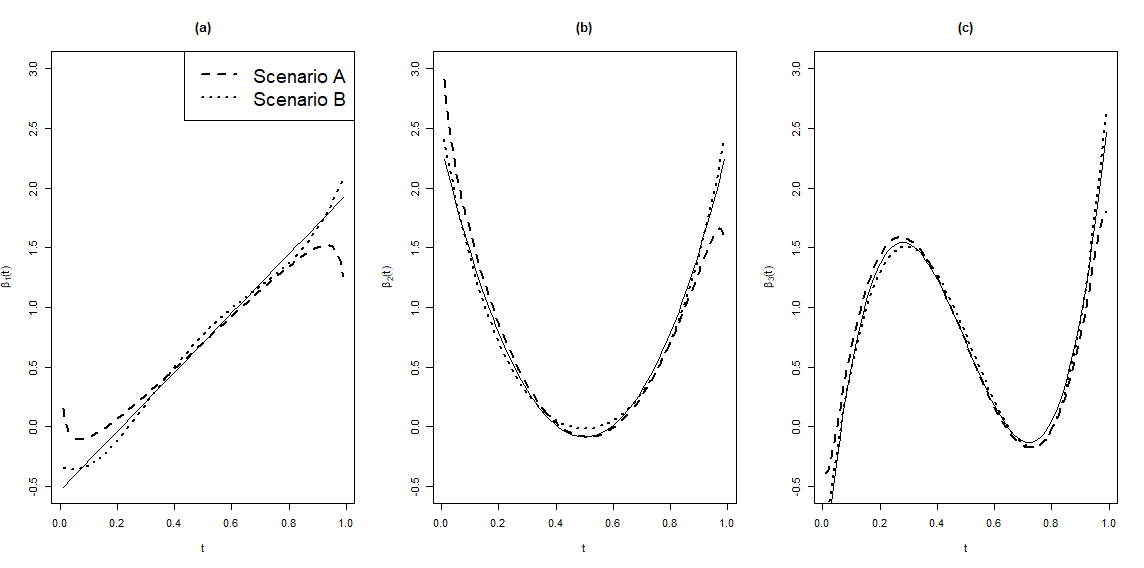}
  \caption{Regression coefficients under scenario A, $\beta_j(t) = \beta_{0,j}(t) + n^{-1/2} \{d_A\delta_{A,j}(t)\}$, and under scenario B, $\beta_j(t) = \beta_{0,j}(t) + n^{-1/2} \{d_B\delta_{B,j}(t)\}$, for (a) $j=1$, (b) $j=2$, and (c) $j=3$, under $n=100$, $d_A=3$, and $d_B=0.6$. The straight lines in each plot represent $\beta_{0,j}(t)$, $j=1,2,3$, respectively.}
  \label{fig:simalt}
\end{figure}
{\color{black}
For each scenario, we apply three incomplete sampling schemes. First, we consider the partially observed functional responses with the random missing period $M_i$, on which functional values on the $i$th trajectory are removed. By following a part of the setting in Remark \ref{rmk:delta}, we generate $M_i = [U_{(p+1)}, U_{(p+k+2)}]$, where $U_{(1)} < \cdots < U_{(2p+k+2)}$ are order statistics of independent random samples of a size $(2p+k+2)$ from Uniform(0,1). We note that $1-\delta(t)$ in Remark \ref{rmk:delta} is set as our indicator process, where employing reversed indicator process does not affect the remarked conclusion. We here set constant parameters $p, k$, as $p=k=3$.  On average, for each simulation set, 30.4 \% of each trajectory is removed by missing interval $M_i$. Second, we consider functional responses irregularly collected over 80 asynchronous grid points with i.i.d. measurement errors following $N(0, 0.5^2$) added to each $Y_i(T_{i,m})$, $m=1, \ldots, 80$. The locations of 80 grid points are uniformly sampled among 100 grids from each observation. Lastly, we consider the partially observed noisy responses collected over irregular grids under the setting in Remark \ref{rmk:delta} with the reserved indicator process specified above. That is, $M_i = [U_{(p+1)}, U_{(p+k+2)}]$, $N_i=60$, and i.i.d. additive measurement errors generated from $N(0, 0.5^2$). Here, locations of 60 grid points are uniformly sampled among available grids on partially sampled trajectories, if there are more than 60 grids on the filtered set. Figure \ref{fig:Y_example} (a) illustrates a randomly selected set of fully observed response trajectories, and three other sets of trajectories in Figure \ref{fig:Y_example} (b),  (c), (d) display partially observed response trajectories filtered by missing random intervals, noisy responses generated over irregular grid points with additive measurement errors, and noisy partially observed responses over irregular grids with additive measurement errors, respectively.
}

\begin{figure}[t!]
  \centering
  \includegraphics[width=4.8in]{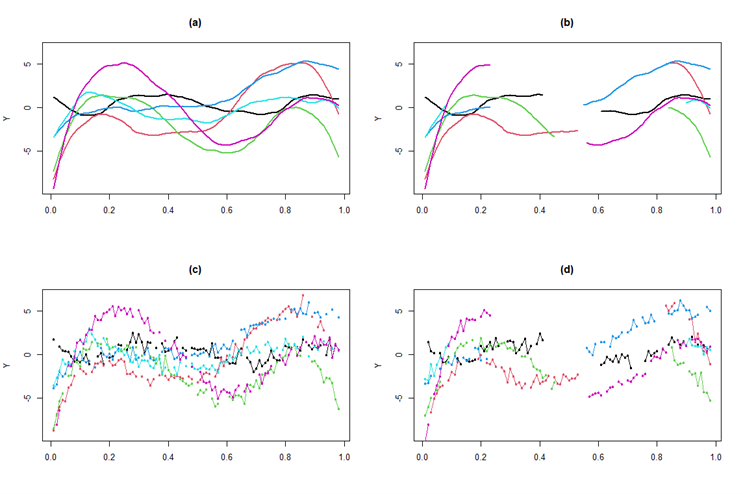}
  \caption{Randomly selected six simulated trajectories of (a) fully observed response data, (b) partially observed response data filtered by independent missing intervals, (c) irregularly observed data with added measurement, and {\color{black} (d) partially observed data over irregular grids with measureme errors.}}
  \label{fig:Y_example}
\end{figure}

\subsection{Empirical size and power}
{\color{black}
We examine the empirical sizes and powers of the proposed procedures for models from fully, partially, irregular, and partially irregular error-prone functional response data using their corresponding test statistics, denoted as $T_n^{\textrm{Full}}$, $T_n$, $T_n^\ast$, and $T_n^{\ast\ast}$ respectively.} Practical implementation steps for each test statistic are provided in the Appendix. All simulation results below were based on 5,000 simulation replicates, and the critical value of the test was estimated by 5,000 bootstrap samples in each simulation run. To calculate the test statistic $T_n^\ast$ and $T_n^{\ast\ast}$ involving kernel smoothing, we chose a common bandwidth that minimizes the leave-one-out cross-validation \citep{wong1983consistency, hardle1985optimal} across all subjects in each simulation sample.
\begin{sidewaystable}
\caption{Empirical size and power at the $5\%$ nominal level for testing $H_0: \beta_j(t) \in \mbox{span}\{ V(4) \}$ under scenario A from fully observed response data ($T_n^{\textrm{Full}}$), partially observed response data ($T_n$), irregularly observed response data with additive measurement errors ($T_n^\ast$), and irregularly observed partial response data with additive measurement errors ($T_n^{\ast\ast}$). 
}
\centering
\vspace{2mm}
\begin{tabular} {cccccccccccccc} 
\hline
 $d_A$ & $n$  & \multicolumn{4}{c}{$\tau=1$} & \multicolumn{4}{c}{$\tau=0.8$} & \multicolumn{4}{c}{$\tau=0.67$}\\
 \cmidrule(lr){3-6} \cmidrule(lr){7-10} \cmidrule(lr){11-14}\\ [-0.1in]
 &  & $T_n^{\textrm{Full}}$ & $T_n$ & $T_n^\ast$ & $T_n^{\ast\ast}$ & $T_n^{\textrm{Full}}$ & $T_n$ & $T_n^\ast$ & $T_n^{\ast\ast}$ & $T_n^{\textrm{Full}}$ & $T_n$ & $T_n^\ast$ & $T_n^{\ast\ast}$ \\
  \hline
\multirow{2}{*}{0} 
    & 100 & 0.060 & 0.068 & 0.061 & 0.072 & 0.065 & 0.060 & 0.064 & 0.072 & 0.061 &  0.050  & 0.066 & 0.071 \\
    & 200 & 0.055 & 0.052 & 0.057 & 0.073 & 0.053 & 0.057 & 0.057 & 0.073 & 0.050 & 0.062  & 0.050 & 0.071 \\
 \hline
\multirow{2}{*}{1} 
    & 100 & 0.075 & 0.067 & 0.067 & 0.070 & 0.082 & 0.81 &  0.076 &  0.074 & 0.107 & 0.085 & 0.096 & 0.078 \\
    & 200 & 0.065 & 0.065 & 0.063 & 0.084 & 0.081 & 0.085 & 0.075 & 0.078  & 0.112 & 0.103 & 0.094 & 0.077\\
 \hline
 \multirow{2}{*}{3} 
    & 100 & 0.153 & 0.101 & 0.118 & 0.086 & 0.354 & 0.288 & 0.21 0 & 0.095 &  0.679 & 0.529 & 0.480 & 0.181 \\
    & 200 & 0.140 & 0.124 & 0.098 & 0.101 & 0.404 & 0.314 & 0.222 & 0.124 &  0.815 & 0.576 & 0.597 & 0.255\\
 \hline
 \multirow{2}{*}{5} 
    & 100 & 0.384 & 0.230 & 0.223 & 0.108 & 0.900 &  0.670 & 0.555 & 0.263 & 1.000 & 0.913 & 0.881 & 0.492\\
    & 200 & 0.398 & 0.258 & 0.207 & 0.111 & 0.958 &  0.752 & 0.634 & 0.349 & 1.000 & 0.996 & 0.900 & 0.608 \\
 \hline
 \multirow{2}{*}{7} 
    & 100 & 0.789 & 0.474 & 0.443 & 0.152 & 0.999 & 0.951 & 0.902 & 0.510 & 1.000 & 1.000 & 0.999 & 0.916 \\
    & 200 & 0.775 & 0.504 & 0.417 & 0.172 & 1.000 & 0.996 & 0.956 & 0.564 & 1.000 & 1.000 & 1.000 & 1.000\\
 \hline
  \multirow{2}{*}{9} 
    & 100 & 0.977 & 0.809  & 0.700 & 0.458 & 1.000  & 1.000 & 0.995 & 0.808 & 1.000 & 1.000 & 1.000 & 0.999 \\
    & 200 & 0.985 & 0.833 & 0.728 & 0.480 & 1.000  & 1.000 & 0.999 & 0.964 & 1.000 & 1.000 & 1.000 & 1.000\\
 \hline
 \label{table:simA}
\end{tabular}
\end{sidewaystable}

Table \ref{table:simA} summarizes results for hypothesis \eqref{hyp:sim} at 5\% nominal level under scenario A from test statistics from corresponding functional response data structures, for $\tau= 1, 0.8, 0.67$. It can be seen that the empirical sizes are reasonably controlled around the nominal level 0.05. Although the sizes under error-prone partially observed structure, corresponding to the test statistic $T_n^{\ast \ast}$, show slightly larger values around 0.07, and it is due to loss of original information with missing intervals and additive noise. 
In terms of power, we investigate the results depending on $\tau$, which regulates the rate that the null-deviated model approaches the null model. As expected, the empirical power increases as $\tau$ decreases or as $d_A$ increases. In addition, the power reasonably approaches to 1 under all settings. Especially for $\tau=1$ of $T_n^{\textrm{Full}}$, the power approaches 1 even with moderate magnitudes of the null-deviated signals, indicating that the condition of $\sum_{m=1}^\infty \pi_m^2 = \infty$ in Theorem \ref{thm-power} is not restrictive in practical application. The relatively deflated powers from $T_n^\ast$ might be due to some loss of the null-deviated signal after applying the smoothing process to noisy data. We observe that the power from $T_n$ goes to 1 with a reasonable but slightly slower rate than $T_n^{\textrm{FULL}}$ shows, and it is from the smaller effective sample sizes at each grid due to partial sampling. Although the lowest powers are observed from $T_n^{\ast \ast}$ under all settings due to most significant loss of original information with missing periods and noisy discretized measurements, we still see the power gradually increases towards 1. {\color{black}Indeed, our extra simulations considering larger values of $d_A$ show that powers under $\tau=1$ from $T_n^\ast$ and $T_n^{\ast\ast}$ become 1 when $d_A=13$ and 17, respectively. }



The simulation results from scenario B are illustrated in Figure \ref{fig:simB}. The results under $n=100$ and $n=200$ are represented by full and dotted lines, respectively. We observe a very similar pattern to that under scenario A with reasonable size controlling at the 0.05 nominal levels and with the behaviors of the power for $d_B>0$. We again confirm that power approaches to 1 when $\tau=1$ under the moderate magnitudes of the null-deviated signals. The power tends to 1 with relatively slower but reasonable rates with an increase of $d_B$ for $T_n$ and $T_n^{\ast}$ due to the same reasons described in results from scenario A. {\color{black} Again, we observe the lowest powers achieved from $T_n^{\ast \ast}$ under $d_B >0$, they gradually approaches towards 1. We note that extra simulations considering larger values of $d_B$ show that powers under $\tau=1$ for $T^\ast$ and $T^{\ast\ast}$ are attained as 1 when $d_B=2.1$ and 3, respectively. It implies empirically consistent properties of our proposed tests.}



{
Although we have only illustrated the simulation result for investigating the finite sample performance of $T_n^\ast$ with $N_i=80$ in Table \ref{table:simA} and Figure \ref{fig:simB}, we observed that the power and size of the proposed test are also well achieved with $N_i=60$, where relatively rich response information is available over the domain. However, under the sparse setting, $N_i=10$ or $30$, we observed relatively unsatisfactory results with the finite sample analysis. We note that, in our simulation setting, the null-deviated signals visualized in Figure 1 are quite subtle, with a delicate difference in the trend and visually detectable discrepancies only at boundaries. Hence, we report a limitation of the proposed method for $T_n^\ast$ such that the estimated regression coefficients calculated from noisy functional responses collected over sparse grids may not be able to effectively detect subtle trend differences near the boundary.
}

\begin{figure}[t!]
  \centering
  \includegraphics[width=4.8in]{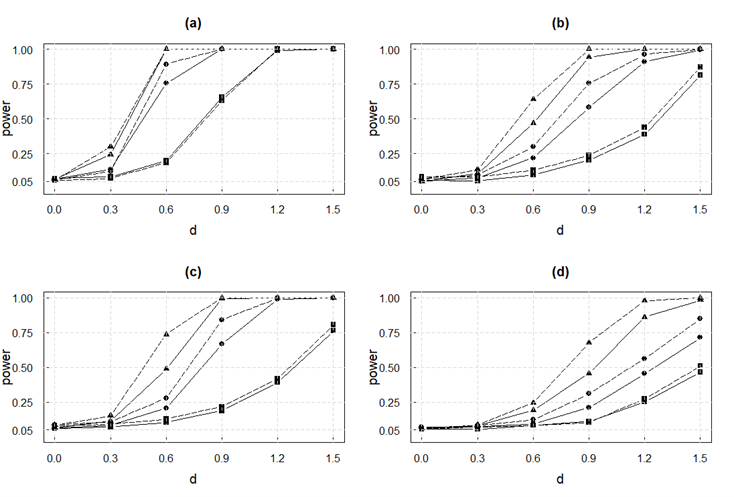}
  \caption{ Empirical size and power at the $5\%$ nominal level for testing $H_0: \beta_j(t) \in \mbox{span}\{ V(4) \}$ under scenario B for (a) fully observed response data, (b) partially observed response data, (c) irregularly observed functional data with additive measurement errors, {\color{black} and (d) irregularly observed partial functional data with additive measurement errors} ($\blacksquare,\tau=1$; $\CIRCLE,\tau=0.8$; $\blacktriangle,\tau=0.67$; $\rule[0.5ex]{1cm}{0.8pt}$, $n=100$; \hdashrule[0.5ex]{1cm}{1pt}{1pt}, $n=200$).} \label{fig:simB}
\end{figure}

\section{Real data application} \label{sec:real-data}

\subsection{The obesity prevalence trend change} \label{subsec:data1}
We illustrate the practical application of the proposed testing procedure through an analysis of the U.S. overweight and obesity prevalence data from 2011 to 2020. It is a part of the data of the U.S. residents regarding their health-related risk behaviors and chronic health conditions, collected by Behavioral Risk Factors Surveillance System (BRFSS) through the state-based telephone interview survey in cooperation with the Centers for Disease Control and Prevention (CDC). The dataset consists of percentages ($\%$) of adults aged 20 and over populations with the weight status of obese, overweight, normal weight, and underweight from 50 states. Along with weight status, socioeconomic status is also measured through educational and income levels of samples. In terms of income, each survey sample is classified into one of five categories; less than \$15,000, \$15,000-\$24,999, \$25,000-\$34,999, \$35,000-\$49,999, and over \$50,000. The full dataset can be found at: https://chronicdata.cdc.gov/Behavioral-Risk-Factors/Behavioral-Risk-Factor-Surveillance-System-BRFSS-P/dttw-5yxu. 
\begin{figure}[t!]
  \centering
  \includegraphics[width=4.8in]{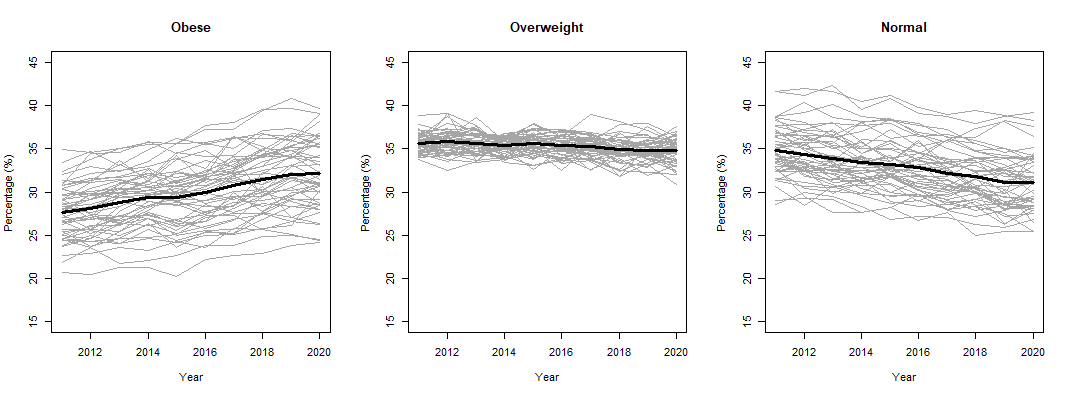}
  \caption{Percentages (\%) of (a) obese, (b) overweight, and (c) normal weight adults among the U.S. adults aged 20 and over populations from 2011 to 2020, from 50 states  (gray lines) and sample means (solid lines) }
  \label{fig:obese}
\end{figure}

Despite growing recognition of the problem, the obesity epidemic continues in the U.S. with steadily rising obesity rates. For example, 1999-2000 through 2017-2018, U.S. obesity prevalence increased from 30.5\% to 42.4\%. Figure \ref{fig:obese} illustrates such trends during recent 10 years from 50 states for obese, overweight, and normal weight groups. The bold lines represent the sample mean trajectories of each group, where its calculation is specified later with the model specification \eqref{intmodel}. With the rising obesity rates, we observe decreasing proportions of normal weight population along with seemingly constant rates of overweight population. We apply the proposed methods to identify shape of the tendency on prevalence rates for each group of the weight status. Furthermore, the gap of obesity prevalence between low and high income groups changes during this time is also examined.


We first investigate the shape of overall prevalence trend for each weight group. Since data is collected over regular grids for all states with a few missing values, we adopt the test statistic $T_n$ for partially observed functional data. Let $Y_i(t_m)$ denote the observed prevalence rate for given weight status group from $i$th state. We formulate the intercept-only model with $\vZ = 0$ in \eqref{fullmodel} and rescaled discrete time points $2011, \ldots, 2020$ to equally spaced $t_m \in [0,1]$,
\begin{equation} \label{intmodel}
    Y_i(t_m) =  \beta(t_m) + \epsilon(t_m),\quad i=1,\ldots,50, \quad m = 1,\ldots, 10.
\end{equation}
Based on it, we obtain the least square estimates $\hat\beta(t_m)= \sum_{i=1}^{50} {Y_i(t_m)}/50$ as the sample trajectories of each weight group, illustrated with bold lines in Figure \ref{fig:obese}. To identify its shape, we consider the null hypotheses for the constant and linear spaces, corresponding to $H_{0,c}: \beta(t) \in \textrm{span}\{V(1)\}$ and $H_{0,l}: \beta(t) \in \textrm{span}\{V(2)\}$, respectively. Here, $V(r) = \{ v_l(t);~t \in [0,1]\}_{l=1}^r$ is an orthonormal set we obtain by applying the Gram-Schmidt process to the polynomial basis $P(r) = \{ t^{l-1}: t\in [0,1]\}_{l=1}^r$, for $r \geq 1$. Table \ref{table:obesity} shows calculated test statistic $T_n$ and corresponding calculated $p$-values for each null hypothesis from each weight group. Calculation details and numerical implementation steps are provided in the Appendix. In Table \ref{table:obesity}, we reject the null hypothesis of constant space $H_{0,c}$ for the obese and normal groups, but not $H_{0,l}$. That is, at a significance level less than $0.001$, obesity prevalence has linearly risen over time, while rates of normal weight population is linearly decreased. On the other hand, we could not find any significant trend as we retain the constant shape hypothesis $H_{0,c}$ at level $0.1$ for the null hypothesis $H_{0,c}$. 

\begin{table}[!t]
\centering
\caption{Calculated test statistic $T_n$ and $p$-values (in parentheses) for null hypotheses of constant and linear trends from each group of weight status }
\vspace{2mm}
\begin{tabular}{cccc}
  \hline
  & Obesity & Overweight & Normal weight \\
  \hline
$H_{0,c}: \beta(t) \in \textrm{span}\{V(1)\}$  & 23.13 ($<0.001) $ & 1.16 (0.162) &  14.74 ($<0.001$) \\
$H_{0,l}: \beta(t) \in \textrm{span}\{V(2)\}$ & 0.34 (0.822) & 0.16 (0.974) & 0.17 (0.967)\\
  \hline
  \label{table:obesity}
\end{tabular}
\end{table}

We next investigate the obesity prevalence over time associated with income levels. In recent literature, statistical analyses on the association between income levels and obesity rates have repeatedly reported that obesity prevalence has been significantly increased at a faster rate mostly in relatively low-income levels \citep{ Cynthia2010, Bently2018, Kim2018}. Figure \ref{fig:obese_income} (a) illustrates obesity prevalence rates for five income levels and their mean trajectories. While all five income levels present increasing obesity prevalence over time, the group of income less than $\$15,000$ shows the highest rates while the groups of income over $\$50,000$ illustrates the lowest rates. We first apply the functional ANOVA to this data, a special case of our proposed testing procedures corresponding to a part of \cite{Zhang2014}. To do this, we formulate the model based on \eqref{fullmodel-re} by setting $(250 \times 4)$ matrix for $\vX = \text{diag}\{\vone_{50}, \ldots, \vone_{50} \}$ and $(250 \times 1)$ vector of $1$'s for $\vZ$. The null hypothesis for fANOVA corresponds to \eqref{gnull}, where $V=\{ 0\}$; i.e., $H_0: \beta_j(t) = 0$, for $t \in [0,1]$, and $j=1,\ldots,4$. By applying the proposed testing procedure, we obtain $p$-value $<0.001$ and conclude that significant differences on obesity rates among different income groups exist. We then apply a type of post hoc test, specifically to examine how the gap of prevalence among lowest highest income group changes over time. Figure \ref{fig:obese_income} (b) shows mean trajectory of gap in obesity rates between the lowest and highest income groups. It is observed that this gap tends to decrease over time and fitted linear and quadratic trending lines are illustrated, respectively. We also consider the fit with the piecewise linear bases, where its fitting details and testing results are specified later.  We apply the proposed procedure based on the test statistic $T_n$ to identify the shape of this gap. Let $\vY(t_m) = \{ \vY_{level1}(t_m)^\top, \vY_{level5}(t_m)^\top \}^\top$, where $\vY_{level1}(t_m)$ is a vector of length $50$ with the elements of obesity rates for the lowest income group from 50 states at $m$th year. Similarly, $\vY_{level5}(t_m)$ denotes a vector for the highest income group. We then specify the model based on \eqref{fullmodel-re} with $\vX=(\vone_{50}, \vzero_{50})^\top$ and $\vZ$ the length 100 vector of $1$'s. Under given model formulation, $\beta(t)$ represents the difference between two groups means. Then two null hypotheses of linear and quadratic functional spaces are considered, $H_{0,l}: \beta(t) \in \textrm{span}\{V(2)\}$ and $H_{0,q}: \beta(t) \in \textrm{span}\{V(3)\}$, respectively.
\begin{figure}[t!]
  \centering
  \includegraphics[width=4.8in]{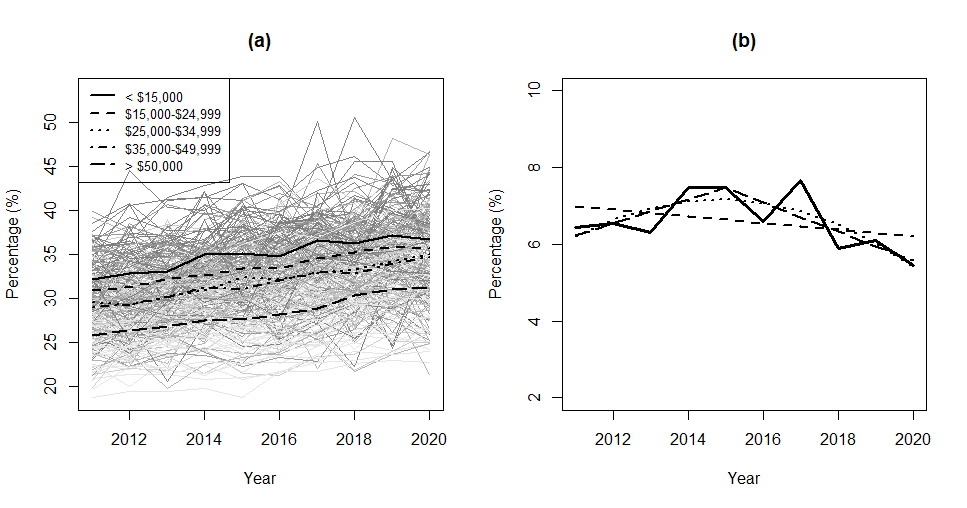}
  \caption{(a) Percentages (\%) of obese prevalence by five income levels from 50 states and (b) mean of differences of the obesity rates between group of income less than $\$15,000$ and group of income over $\$50,000$, with fitted lines using linear bases ($\hdashrule[0.5ex]{1.5cm}{1pt}{1.5mm 2pt}$), using quadratic bases ($\hdashrule[0.5ex]{1cm}{0.5pt}{1.5pt}$) and piecewise linear bases ($\hdashrule[0.5ex]{2cm}{1pt}{3.5mm 2pt})$.}
  \label{fig:obese_income}
\end{figure}
By applying the proposed testing procedures; for the test under $H_{0,l}$, we obtain $T_n = 10.75$ with the $p$-value $0.02$; and for the test under $H_{0,q}$, $T_n = 5.03$ and $p$-value is $0.23$. Under significance level $\alpha = 0.05$, we fail to reject the quadratic null space and conclude that the gap of obesity prevalence between lowest and highest income groups is significantly decreasing with the quadratic shape. To demonstrate the further application of our method with null hypothesis with other types of bases, we try the hypothesis testing for $H_{0,pl}: \beta(t) \in \textrm{span}\{U(3)\}$, where $U(3)$ represent a set of three orthonormal B-spline bases derived from the piecewise linear functions with knots at $0, 0.5,$ and 1, where the internal knot 0.5 is chosen by the estimated peak from the quadratic fit. Under this null hypothesis, we obtain $T_n = 4.70$ and $p$-value 0.28. By comparing obtained p-value 0.28 with  the $p$-value $0.23$ derived from the null hypothesis $H_{0,q}$, we observe slightly stronger statistical evidence on the conclusion for the piecewise linear shape on gap between two groups under given sample sizes. We note that results from smoothed trajectories through the test statistic $T_n^*$ leads the same inferential conclusions for $H_{0,l}$, $H_{0,q}$, and $H_{0,pl}$, although they are not presented here. It empirically demonstrates the performance of our proposed method in detecting significant functional shape even under non-smoothed raw trajectories.

\subsection{Human motion analysis in ergonomics}

We illustrate another data example in automotive ergonomics, previously analyzed by \cite{faraway1997regression}, \cite{shen2004f}, \cite{zhang2011statistical}, \cite{chen2020model}, and among others. The Center for Ergonomics at the University of Michigan collected data on body motions of an automobile driver. As part of the project, the right elbow angles of the test driver were captured as time-varying responses when the driver's hand leaves the steering wheel until reaching 20 different locations in the car. There were 3 repeated reaches to each of the different targets located near the glove compartment, headliner, radio panel, and gear shifter. 

{\color{black}  We associate observed discrete trajectories of elbow angles $R_{ij}(t_{ij,m})$ with with the $(x,y,z)$-coordinate of a reaching target with extra variables as
\begin{equation}
\begin{aligned}
    R_{ij}(t_{ij,m})
    &= \mu_0(t) 
    + \sum_{k=1}^3 \alpha_k(t_{ij,m}) d_{ik}  
    + \sum_{l=1}^3 \beta_l(t_{ij,m}) c_{il} \\
    &\qquad + \sum_{k=1}^3\sum_{l=k}^3 \gamma_{kl}(t_{ij,m}) c_{ik} c_{il}
    + \varepsilon_{ij}(t_{ij,m}),
\end{aligned} \label{faraway-reg-model}
\end{equation}
for $i=1, \ldots, 20$, $j=1,2,3$, and $k=1, \ldots, N_i$, }
where $(c_{i1},c_{i2},c_{i3})$ represents the $(x,y,z)$-coordinate of a target location with its origin at the initial hand posture on the steering wheel and  $d_{ik}$'s are $0$-$1$ dummy variables indicating four nominal areas of different targets. Specifically, $d_{i1} = 1$ if the target is located near the headliner, $d_{i2} = 1$ if the radio, $d_{i3} = 1$ if the gear shifter, and zeros otherwise so that we set the glove compartment for the baseline location. By adding the nominal target information to the conventional model, we are able to statistically compare the changes of elbow angles from different experimental conditions. Among $60$ experiments, we drop one trial which has been excluded in the literature, where the researchers revealed that the driver's motion was mistaken while reaching the target. See \cite{faraway1997regression} for more details about the experimental settings.

Since observed discrete trajectories of elbow angles reveal some noises due to the measurement errors, \cite{faraway1997regression} applied the smoothing splines to raw data to respect the smoothness of human motion and obtained pre-smoothed angle random curve. We denote it as $\tilde{R}_{ij}^\ast(t)$, where the tracking time points $t_{ij,1}, \ldots, t_{ij,N_{ij}}$ were re-scaled to $[0,1]$ for each of 60 reaches. The pre-smoothed random sample $\{ \tilde{R}_{ij}^\ast: i=1, \ldots, 20, \, j=1,2,3\}$ can be analyzed by the standard one-way functional ANOVA. 
 We note that the model \eqref{faraway-reg-model} turns out to be adequate for the data as the bootstrap-based test \citep{faraway1997regression} does not reject the lack of fit compared with the functional ANOVA model ($\textrm{p-value}=0.436$).  \cite{chaffin2002simulating, chaffin2005improving} also considered similar approaches to the driver's motion prediction in a larger dataset by adding extra variables to statistically control different experimental conditions. 

 In this example, we aim to analyze the shape of the time-varying motion changes rather than find a predictive model for an arbitrary target location. 
 {\color{black}
Based on the asymptotic equivalence between splines and certain class of kernel estimates \citep{Silverman1984, Lin2004}, we apply the proposed method to the pre-smoothed random sample $\{ \tilde{R}_{ij}^\ast: i=1, \ldots, 20, \, j=1,2,3\}$ to test the null hypothesis $H_{0}^\alpha:$ $\{ \alpha_1, \alpha_2, \alpha_3 \}$  $\in$ $\mathrm{span}\{V(2)\}$ with the same $V(2)$ defined in Section \ref{subsec:data1}. We perform inference using $T_n^\ast$ from Section \ref{sec:kernel} and find that the null hypothesis cannot be rejected ($T_n^\ast=3.16$, \,$\textrm{p-value}=0.660$).} This result together with Figure \ref{fig1:faraway} shows that, compared to the glove reaching experiment, the driver stretched their elbow less and moved slower at a constant relative angular velocity when reaching different area. We also individually test several hypotheses such as $H_{0}^\beta: \{ \beta_1, \beta_2, \beta_3 \} \in \mathrm{span}\{ V(2) \}$ ($T_n^\ast=3.88$, \,$\textrm{p-value}=0.930$), $H_{0}^{\gamma_{kk}}: \{ \gamma_{11}, \gamma_{22}, \gamma_{33} \} \in \mathrm{span}\{ V(2) \}$ ($T_n^\ast=2.86$, \,$\textrm{p-value}=0.710$), and $H_{0}^{\gamma_{kl}}: \{ \gamma_{12}, \gamma_{13}, \gamma_{23} \} \in \mathrm{span}\{ V(2) \}$ ($T_n^\ast=4.68$, \,$\textrm{p-value}=0.409$). We close this section by reporting that all twelve hypotheses we have tested were still not rejected after applying the multiple comparison adjustment, both the Bonferroni and Benjamini-Hochberg corrections, at $5\%$ significance level, implying statistically significant linear trends on them. 

\begin{figure}[!t]
  \centering
  \includegraphics[width=1.55in]{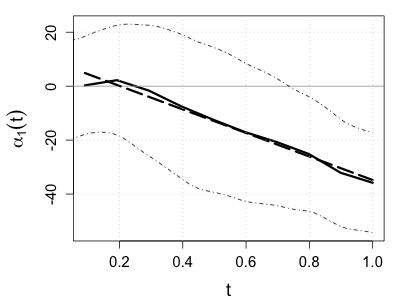}
  \includegraphics[width=1.55in]{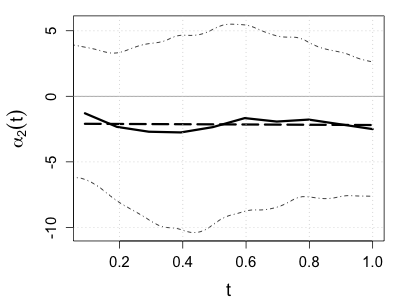}
  \includegraphics[width=1.55in]{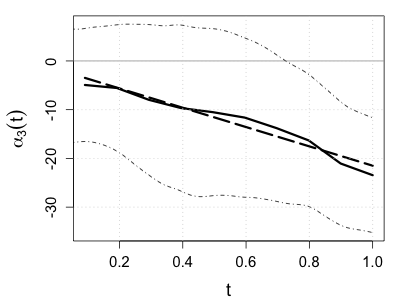}\\
  \includegraphics[width=1.55in]{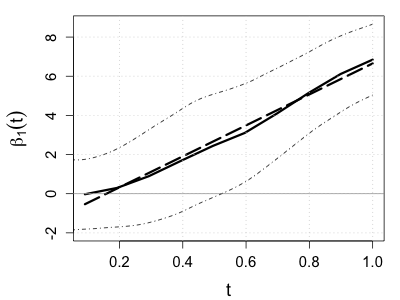}
  \includegraphics[width=1.55in]{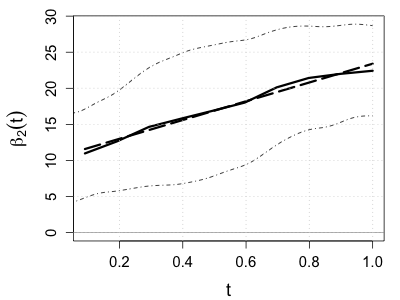}
  \includegraphics[width=1.55in]{faraway-example-coef-x.png}\\
  \includegraphics[width=1.55in]{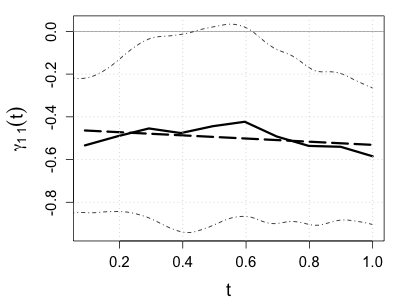}
  \includegraphics[width=1.55in]{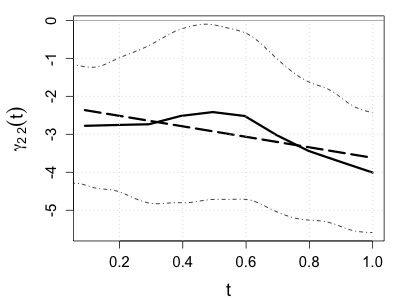}
  \includegraphics[width=1.55in]{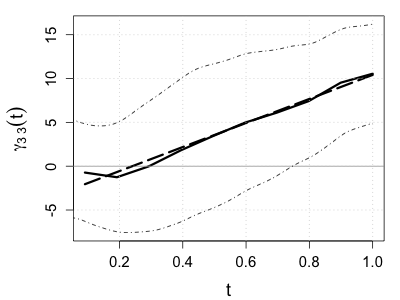}\\
  \includegraphics[width=1.55in]{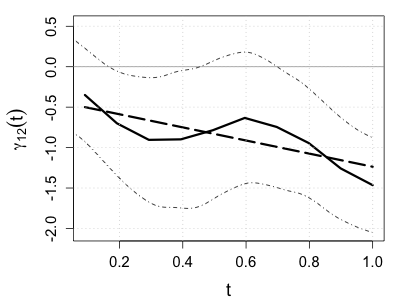}
  \includegraphics[width=1.55in]{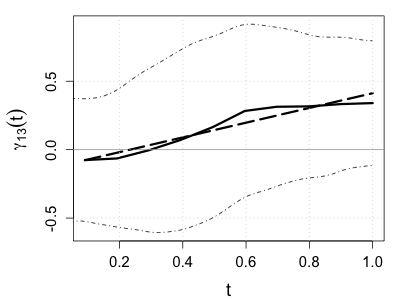}
  \includegraphics[width=1.55in]{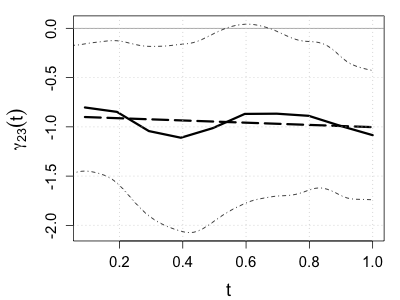}\\
  \caption{The regression coefficient estimates of the model \eqref{faraway-reg-model} are depicted. The solid and dot-dashed lines are coefficient estimates and their $95\%$ confidence bands, respectively. The long-dashed lines show the estimates under the null hypotheses $H_{0,l}$.}
  \label{fig1:faraway}
\end{figure}

\section{Discussion} \label{sec:discussion}
We have presented a statistical procedure for testing shape-constrained hypotheses on regression coefficients in function-on-scalar regression models, generalizing existing methods such as fANOVA that consider nullity hypotheses only. The approach presented here enables inferences about temporal/spatially varying coefficient effects as well. The large sample properties of the proposed test were investigated by deriving the asymptotic null distribution of the test statistic and consistency of the test against local alternatives. The methodology was demonstrated under three incomplete sampling situations; (i) partially observed, (ii) irregularly observed error-prone, and (iii) composition of former two incomplete functional response data. A few studies have recently illustrated goodness-of-fit tests for functional linear models under fully observed responses, but handling incomplete sampling designs was not studied either in theory or practice. Furthermore, the critical value in our methodology can be approximated with the spectral decomposition of covariance function, for which one can easily exploit the existing methods in the recent developments in functional data analysis. A key aspect of the methodology developed here is the specification of a relevant shape hypothesis of interest. Ideally, the application defines the relevant shape space. Otherwise, we can use standard curve-fitting hypotheses defined by, for example, polynomial basis functions, exponential functions, or periodic functions cycling at different frequencies.

In Section \ref{sec:kernel}, we considered functional data, where each sample path is observed on randomly spaced discrete points of size $N_i$. Assuming that $N_i$'s are increasing as the sample size increase, which is often called ``densely observed'' functional data, we adopted the individual smoothing strategy as interpolation. Another interesting and challenging situation is when the functional data are so sparsely observed that the individual smoothing strategy employed here is not effective. In this case, one may consider a functional principal components based approach to reconstruct individual curves. Recently, \cite{kneip2020optimal} proposed optimal reconstruction of individual curves in which each of the incomplete $n$ functions is observed at discrete points considerably smaller than $n$ in finite sample analysis. They showed that the functional principal components based approach can provide better rates of convergence than conventional smoothing methods, where $\min_i \{N_i\} \asymp n^\theta$ as $n \to \infty$ for some $\theta > 0$. However, hypothesis testing under the functional principal component analysis framework remains to be developed. 

\appendix




\section{Technical Details}

\subsection{Numerical Implementation}
We first present the numerical implementation of the proposed test for the fully observed response data. In practice, the response $Y_i(t)$ is collected in a discrete manner over a dense grid $t_1, \ldots, t_{N_i}$. For simplicity, we focus on the case $N_i = N$ and all the individual functions are observed at a common grid of design time points. If the design time points are different for different individual trajectories, we can apply the kernel smoothing to obtain the evaluations at a common grid points under its uniform consistency property, demonstrated in Section \ref{sec:kernel}.

Suppose that two design matrices $\vX$ and $\vZ$ are orthogonalized as in \eqref{fullmodel-re} and a set of orthonormal bases \{$v_l$; $l=1,\ldots, r\}$ is given for the null hypothesis \eqref{gnull}. We calculate $(p \times N)$ matrix $\hat {\boldsymbol{\beta}} = (\hat{\boldsymbol{\beta}}^\top_1, \ldots, \hat{\boldsymbol{\beta}}^\top_p)^\top$, where $\hat {\boldsymbol{\beta}}_j = \{\hat\beta_j(t_1), \ldots, \hat\beta_j(t_N)\}^\top$ is the least square estimator of $\vbeta_j$ at each grid. The test statistic $T_n^{\text{Full}}$ is obtained based on $D = \hat {\boldsymbol{\beta}} - \mathcal{L} \hat {\boldsymbol{\beta}}$, where $(p \times N)$ matrix $D=(D_1, \ldots, D_N)$ is defined with length $p$ vector $D_m$, $m=1\ldots, N$, and consists of $j$th row  representing the length $N$ regression residuals by fitting the linear regression for the response $\boldsymbol{\hat\beta}_j$ with $r$-columns of matrix V as covariates. Here, each column of matrix V are discretized orthonormal bases $v_1, \ldots, v_r$ evaluated at $N$ grid points. Then, the test statistic $T_n^{\text{Full}}$ is approximated by $N^{-1}\sum_{m=1}^N D_m^\top(\tilde{\mathbb{X}}^\top \tilde{\mathbb{X}}) D_m$. We next find the empirical critical value for the level $\alpha$ test. We calculate the $(N \times N)$ covariance matrix of residuals, denoted as $\Gamma = [\gamma_{m m'}]_{1 \leq m, m' \leq N}$, based on $ {\boldsymbol{r}}_i= \{r_i(t_1), \ldots, r_i(t_N)\}^\top$, $i=1,\ldots, n$, where $r_i(t_m) ={Y_i}(t_m) - \tilde{\mathbb{X}} \hat\vbeta(t_m) - \mathbb{Z} \hat\veta(t_m) $ under  $\hat\veta(t_m)=(\mathbb{Z}^\top \tilde{\mathbb{Z})}^{-1} \mathbb{Z}^\top {\vY}(t_m)$. We then derive the discretized $\tilde \gamma(s,t)$ following the formula in  \eqref{gp-var}, denoted as $\tilde \Gamma $, by calculating $\tilde \Gamma = \Gamma -\Gamma_{(c)} - \Gamma_{(r)} + \Gamma_{(c,r)}$, where $\Gamma_{(c)}=(\hat{ \boldsymbol{\gamma}}_{(c)1}, \ldots, \hat{\boldsymbol{\gamma}}_{(c)N})^\top$ is the matrix of the fitted multi-response regression values based on $N$ separate regressions, with each column of $\Gamma$ as the response, and the $r$-columns of $V$ as covariates. Simiarly, $\Gamma_{(r)}$ is the matrix of fitted multi-response regression values, with each row of $\Gamma$ as the response, and the $r$-columns of $V$ as covariates. Lastly, $\Gamma_{(c,r)}$ is the matrix of fitted values by applying previous two steps to $\Gamma$ subsequently. We then generate a large bootstrap samples of $\hat T_0 = \sum_{k=1}^{\hat K} \hat\lambda_k A_k $,  $A_k \stackrel{i.i.d.}{\sim} \chi^2_p ,$ where $\hat K$ denotes the number of positive eigenvalues of $\tilde \Gamma$, and use its $(1-\alpha)\%$ quantile as the critical value.


To perform the proposed hypothesis testing under irregularly collected response data with additive measurement errors, we replace $Y_i(t)$ by kernel smooth estimates of $\tilde Y_i^*(t)$, obtained by \eqref{kernsmooth}, and apply the procedures described for the fully observed response data. To find the optimal smooth parameter in kernel estimation, one may adopt the leave-one-out cross-validation.

For the application of the proposed testing procedure to partially observed functional response data, we calculate $(p \times N)$ matrix $D^w = \hat{\boldsymbol{\beta}}^w - \mathcal{L} \hat{\boldsymbol{\beta}}^w$, where $D^w = (D_1^w, \ldots, D_N^w)$ and $\hat {\boldsymbol{\beta}}^w_j = (\hat\beta_j^w(t_1), \ldots, \hat\beta_j^w(t_N))^\top$, and approximate $T_n =N^{-1}\sum_{m=1}^N {D^w_m}^\top(\tilde{\mathbb{X}}^\top \tilde{\mathbb{X}}) D^w_m$. To find the empirical critical value for the level $\alpha$ test, we estimate $\Gamma$ by employing nonparametric covariance surface estimation method applicable to sparse functional data, available through the function \texttt{GetCovSurface} in R package `fdapace'. The optimal bandwidth for the surface estimation can be found through the cross-validation steps. Next, we calculate $\hat v(t_m,t_{m'}) = \sum_{i=1}^{n} \delta_i(t_m) \delta_i(t_{m'})/n,$ and $\hat b(t_m) = \sum_{i=1}^{n} \delta_i(t_m)/n$. Then $(N \times N)$ matrix $\Xi$, discretized version of $\vartheta(s,t)$ in Theorem \ref{cor:regression}, can be derived, where its $(m, m')$-th element is calculated as $\Xi_{mm'}=\Gamma_{mm'} \Pi_{mm'}$. Here, $(N \times N)$ matrix $\Pi$ has its $(m,m')$-th element as $\Pi_{mm'}=\hat v(t_m, t_{m'})\hat b(t_m)^{-1} \hat b(t_{m'})^{-1}.$ We next calculate  $\tilde \Xi = \Xi -\Xi_{(c)} - \Xi_{(r)} + \Xi_{(c,r)}$, where $\Xi_{(c)}$ and $\Xi_{(r)}$ are obtained by following definitions of each term, described in the implementation for the fully observed data. In practical application, we adopt the standardized test statistic ${\breve{T}}_n$ $=$ $N^{-1}\sum_{m=1}^N$ ${\breve{D}}_m^{w \top}(\tilde{\mathbb{X}}^\top \tilde{\mathbb{X}}) {\breve{D}}_m^{w}$, where ${\breve{D}}_m^{w} = D_m^w \hat b(t_m) \hat v(t_m, t_m)^{-1/2}$ and obtain an approximate critical value from $\tilde \Xi^* = \Xi^* -\Xi_{(c)}^* - \Xi_{(r)}^* + \Xi_{(c,r)}^* $, where $ \Xi^*_{mm'} = \Gamma_{mm'} \Pi^*_{mm'}$ with $\Pi^*$ representing the standardized matrix of $\Pi$ having unit variance for diagonals. The standardized testing procedure empirically shows the improved performance in size controlling in simulation studies. 

Lastly, we can calculate $T_n^{\ast \ast}$ by combining two previous steps, smoothing process over observed trajectories and calculation of test-statistic under partial structures. 

\subsection{Technical Details for Section \ref{subsec:partial}}

\subsubsection*{Proof of Theorem \ref{cor:regression}}

Suppose $\E(\tilde{\mathbb{X}}| \mathbb{Z}) = \vzero$ without loss of generality. Under the null hypothesis,
\begin{align}
    \begin{split}
        \hat \vbeta^w(t) 
        &= (\tilde{\mathbb{X}}^\top \mathbb{W}(t) \tilde{\mathbb{X}})^{-1} \tilde{\mathbb{X}}^\top \mathbb{W}(t) \{\tilde{\mathbb{X}}^\top \vbeta_0(t) + \mathbb{Z}^\top \boldsymbol{\eta}(t)  + \vepsilon(t)\} \\
        &= \vbeta_0(t)+  (\tilde{\mathbb{X}}^\top \mathbb{W}(t) \tilde{\mathbb{X}})^{-1} \tilde{\mathbb{X}}^\top \mathbb{W}(t) \vepsilon(t),
    \end{split} \nonumber
\end{align}

and let $\vZ_n(t) = \sqrt{n} (\hat\vbeta^w(t) - \vbeta_0(t)) $. Then we can write 
\begin{equation}
\begin{aligned}
\vZ_n(t) &= \Big( \frac{\tilde{\mathbb{X}}^\top \mathbb{W}(t) \tilde{\mathbb{X}}}{nb(t)} \Big)^{-1} \frac{\sqrt{n}\tilde{\mathbb{X}}^\top \mathbb{W}(t) \vepsilon(t)}{nb(t)} \nonumber\\
 &= \frac{nb(t)}{\sum_{i=1}^{n} \delta_i(t)}\Big( \frac{\tilde{\mathbb{X}}^\top \mathbb{W}(t) \tilde{\mathbb{X}}}{\sum_{i=1}^{n} \delta_i(t)}\Big)^{-1}  \frac{\sqrt{n}\tilde{\mathbb{X}}^\top \mathbb{W}(t) \vepsilon(t)}{nb(t)} \label{beta-exp}.
\end{aligned}
\end{equation}
Let $\tilde{\mathbb{X}}=(\tilde{\mathbb{X}}_1, \ldots,\tilde{\mathbb{X}}_n)^\top$, where $\tilde{\mathbb{X}}_i = (\tilde{x}_{i1}, \ldots, \tilde{x}_{ip})^\top$, and let 
\begin{equation}
    \vV_n(t) = n^{-1/2} \tilde{\mathbb{X}}^\top \mathbb{W}(t) \vepsilon(t)/b(t)
\end{equation} 
be the $p$-variate random functions with the zero mean, corresponding to the third term in \eqref{beta-exp}. Its $j$-th element is specifically written as $ n^{-1/2} \sum_{i=1}^n \tilde{x}_{ij} \delta_i(t) \epsilon_i(t) / b(t)$. Let ${\mathbb{V}}_n = (\vV_n(t_1), \ldots, \vV_n(t_Q))$, where $\mathcal{T}_Q = \{t_q \in [0,1]: q=1, \ldots, Q\}$ is a finite collection of any $Q$ time points, for $Q \geq 1$. By the multivariate CLT and the mutual independence among $\tilde x_{ij}$, $\delta_i$, and $\epsilon_i$, we have
\begin{equation} \nonumber
\mathrm{vec}({\mathbb{V}}_n) = \big( \vV_n(t_1), \ldots, \vV_n(t_Q) \big)^\top \stackrel{d}{\to} MVN(\boldsymbol{0}_{pQ}, \Xi \otimes \Psi),
\end{equation}
where $\Xi = \big[\vartheta_{q q'}\big]_{1 \leq q,q' \leq Q}$ is the $Q \times Q$ covariance matrix with 
\[\vartheta_{qq'} = \gamma(t_q, t_{q'}) v(t_q, t_{q'}) b(t_q)^{-1} b(t_{q'})^{-1},\] 
$\Psi = \big[ \Psi_{jj'}\big]_{1 \leq j,j' \leq p}$ is the $p \times p$ matrix with $\Psi = E(\Var(\boldsymbol{X} | \boldsymbol{Z} ))$, and the Kronecker product of $\Xi$ and $\Psi$ is given by
	\begin{equation} \nonumber
		\Xi \otimes \Psi
		=
		\left[
		\begin{array}{ccc}
			\vartheta_{11} \Psi & \cdots & \vartheta_{1Q} \Psi\\
			\vdots	&	\ddots	&	\vdots\\
			\vartheta_{Q1} \Psi & \cdots & \vartheta_{QQ} \Psi
		\end{array}
		\right]
		\in \mathbb{R}^{(pQ) \times (pQ)}.
	\end{equation}
We specifically derive $\Xi \otimes \Psi$ as follows. For $p$-variate random variable $\vV_n(t_q)$, the diagonal of its asymptotic covariance matrix, i.e., $(j,j)$-th element of the matrix, is derived as  $\gamma(t_q, t_q)b(t_q)^{-1} E (\tilde x_{ij}^2) = \vartheta_{qq} \Var(\tilde x_{ij}) $, and the $(j,j')$-th element of the covariance matrix, for $j \neq j'$, is $\gamma(t_q, t_q)b(t_q)^{-1} E (\tilde x_{ij}  \tilde x_{ij'}) = \vartheta_{qq} \Cov(\tilde x_{ij}, \tilde x_{ij'}) $. That is, the block diagonal covariance matrix of $\mathrm{vec}({\mathbb{V}}_n)$ is $\vartheta(t_q, t_q) \Psi$. We then examine the block off-diagonal covariance matrix of  $\mathrm{vec}( {\mathbb{V}}_n)$ by calculating the covariance between $\vV_n(t_q)$ and $\vV_n(t_{q'})$, for $q \neq q'$. we can show that the diagonal $(j,j)$-th element of the covariance matrix is $\vartheta_{qq'}\Var(\tilde{x}_{ij})$ and the $(j,j')$-th element of the matrix, for $j \neq j'$, is $\vartheta_{qq'}\Cov(\tilde{x}_{ij}, \tilde{x}_{ij'})$. That is, $p \times p$ off-diagonal block covariance matrix of $\mathrm{vec}({\mathbb{V}}_n)$ is written as $\vartheta_{qq'} \Psi$. 
By \cite{gupta2018matrix} and \cite{chen2020multivariate}, the multivariate process $\{ \vV_n(t) : t \in [0,1] \}$ converges to the multivariate Gaussian process in distribution as
\[
\{ \vV_n(t) : t \in [0,1] \} \stackrel{d}{\to} GP_p (\vzero_{p}, \vartheta \Psi),
\]
where the finite-dimensional restrictions of $\vartheta$ is given by the covariance matrix $\Xi$. Next, we can show that the $(p \times p)$ matrix $\tilde{\mathbb{X}}^\top \mathbb{W}(t) \tilde{\mathbb{X}}/ \sum_{i=1}^{n} \delta_i(t)$ in the second term of \eqref{beta-exp} converges to $\Psi$ in probability, under the conditions C2 and C4. Let $\tilde{\vV}_n(t) = (\tilde{\mathbb{X}}^\top \mathbb{W}(t) \tilde{\mathbb{X}}/ \sum_{i=1}^{n} \delta_i(t))^{-1}$ $\vV_n(t)$, then $\tilde\vV_n(t) \stackrel{d}{\to} GP_p(\vzero_p, \vartheta \Psi^{-1})$, for $t \in [0,1],$ by the Slutksy's lemma. Note that
\[
\sup_{t \in [0,1]}  \left|  \tilde\vV_n(t) - \vZ_n(t) \right| \leq \sup_{t \in [0,1]} |\tilde\vV_n(t) | \cdot \sup_{t \in [0,1]} \left|  1 - \frac{nb(t)}{\sum_{i=1}^{n} \delta_i(t)}\right|,
\]
where $\sup_{t \in [0,1]}|\vZ(t)| \triangleq \sup_{t \in [0,1]} \sup_{j \in \{1,\ldots, p \}} Z_j(t)$, for $p$-variate random functions $\vZ(t)=(Z_1(t),$ $\ldots,$ $Z_p(t))^\top$. Following the similar lines of the proof of Theorem 4 and the Lemma 2.1 provided in \cite{Park2021}, we have
\[
\sup_{t \in [0,1]}  \left|  \tilde\vV_n(t) - \vZ_n(t) \right| = O_p(n^{-1/2}).
\]
Then Corollary \ref{cor:regression} is an immediate consequence of Slutksy's lemma.

\subsubsection*{Proof of Theorem \ref{thm-alternative-dist} and Corollary \ref{thm-power}}
We first present the proof of Theorem \ref{thm-alternative-dist}.
Following the similar arguments used in Theorem 1 by  \cite{zhang2011statistical}, we have
\begin{equation}
\begin{aligned}
	T_n
	&= \sum_{j=1}^p \int_0^1 W_j(t)^2 \, \mathrm{d}t + o_P(1)\\
	&= \sum_{j=1}^p \sum_{m=1}^{\infty} \psi_{jm}^2 + o_P(1),
\end{aligned}
\end{equation}
where $\mathbf{W} = (W_1, \ldots, W_p)^\top \sim \textrm{GP}_p( \tilde{\boldsymbol\Delta}, \tilde\vartheta \mathbb{I}_p)$. The eigen-decomposition of $\tilde\vartheta(s,t)$ leads to $W_j(t) = \sum_{m=1}^{\infty} \psi_{jm} \phi_m(t)$, where the series converges in $L^2$, uniformly for $t \in (0,1)$, and $\psi_{jm} = \langle W_j, \phi_m \rangle  \sim N(\langle \tilde\Delta_j, \phi_m \rangle, \lambda_m)$ independent for all $j=1, \ldots, p$ and $m \geq 1$. 
Since $\| \tilde\Delta_j \|_2^2 = \sum_{m=1}^\infty |\langle \tilde\Delta_j, \phi_m \rangle|^2 < \infty$ 
it follows that 
\[ 
\sum_{m=1}^\infty \textrm{Var} \big(\psi_{jm}^2\big)  = \sum_{m=1}^\infty 2\lambda_m\big(1 + 2 |\langle \tilde\Delta_j, \phi_m \rangle|^2/\lambda_m\big) < \infty
\]
for all $j=1, \ldots, p$. Therefore,
\begin{equation}
	T_\Delta
	\stackrel{a.s.}{=} \sum_{m=1}^{\infty} \sum_{j=1}^p \psi_{jm}^2
	\stackrel{d}{=}
	\sum_{m=1}^\infty \lambda_m B_m,
\end{equation}
where $B_m = \sum_{j=1}^p \psi_{jm}^2/\lambda_m$ has the non-central $\chi^2$-distribution with $p$ degrees of freedom and the non-central parameter $\kappa_m^2 = \pi_m^2/\lambda_m$. Since $W_1, \ldots, W_p$ are independent Gaussian processes, $B_1, B_2, \ldots$ are independent. The proof of Corollary \ref{thm-power} case (i) is a special case with non-centrality parameter on $\chi^2$ distribution with $p$ degrees of freedom. 

\subsubsection*{Proof of Corollary \ref{thm-power} case (ii)}
	The proof with $\tau \in [0,1)$ follows from Theorem \ref{thm-alternative-dist} as $\Psi^{1/2}  (\boldsymbol{\mathcal{I}} - \boldsymbol{\mathcal{L}})(n^{-\tau/2}\boldsymbol\Delta) \to \infty$ as $n \to \infty$. When $\tau=1$, we assume that $\sum_{m=1}^\infty \pi_m^2 = \infty$. Let $\zeta_{jm}$ denote a standard normal random variable independent for all $j=1, \ldots, p$ and $m \geq 1$. We note that
	\begin{equation}
	\begin{aligned}
		B_m
		&= \sum_{j=1}^p \psi_{jm}^2/\lambda_m\\
		&\stackrel{d}{=} \sum_{j=1}^p \zeta_{jm}^2 + 2 \sum_{j=1}^p \zeta_{jm} \langle \tilde\Delta_j, \phi_m \rangle/\sqrt{\lambda_m} + \sum_{j=1}^p |\langle \tilde\Delta_j, \phi_m \rangle|^2/\lambda_m\\
		&\stackrel{d}{=} A_m + 2 \tilde\rho_m \zeta_{1m} + \pi_m^2/\lambda_m,
	\end{aligned}
	\end{equation}
	where $\tilde\rho_m = \sum_{j=1}^p \langle \tilde\Delta_j, \phi_m \rangle/\sqrt{\lambda_m}$ for $m \geq 1$ with $A_m$ and $B_m$ defined in the previous theorems. It follows from Corollary \ref{thm-power} case (i) and \eqref{thm-alternative-dist-eq} that
	\begin{equation}
	\begin{aligned}
		\lim_{n \to \infty} P(T_\Delta \geq t_\alpha | H_{1n})
		&= P\bigg( \sum_{m=1}^\infty \lambda_m B_m \geq t_\alpha \bigg)\\
		&= P\bigg( T_0 + 2 \sum_{m=1}^\infty \lambda_m \tilde\rho_m \zeta_{1m} + \sum_{m=1}^\infty \pi_m^2 \geq t_\alpha \bigg).
	\end{aligned}
	\end{equation}
	Let $\Pi^2 = \sum_{m=1}^\infty \lambda_m \pi_m^2$.
	We note that 
 \begin{equation}
 \begin{aligned}
     \sum_{m=1}^\infty \textrm{Var}(\lambda_m\tilde\rho_m \zeta_{1m})
     &= \sum_{m=1}^\infty \lambda_m^2 \bigg( \sum_{j=1}^p \langle \tilde\Delta_j, \phi_m \rangle/\sqrt{\lambda_m} \bigg)^2\\
     &\leq p^2 \sum_{m=1}^\infty \lambda_m \pi_m^2 = p^2 \Pi^2,
 \end{aligned}
 \end{equation}
	where $\Pi^2 \leq \lambda_1 \sum_{m=1}^\infty \pi_m^2 = \lambda_1 \sum_{j=1}^p \| \tilde\Delta_j \|_2^2 < \infty$. Therefore, we can write $\sum_{m=1}^\infty \lambda_m \tilde\rho_m \zeta_{1m} \stackrel{d}{=} \Pi Z_0$, where $Z_0 \sim N(0,1)$ is independent of $T_0$. This completes the proof.


\subsection{Technical Details for Section \ref{sec:kernel}}

\begin{lem} \label{lem-unif-conv}
	Let $\eta_1(t), \ldots, \eta_N(t)$ be independent and random functions such that there exists $B_N > 0$ satisfying
	$\max_{1 \leq j \leq N} E \|\eta_j\|_\infty^k = O(B_N)$ for some $k > 2$ and
	$\max_{1 \leq j \leq N} \mathrm{Lip}(\eta_j) = O_P(1)$, where $\mathrm{Lip}(f)$ denotes the Lipschitz constant of $f$.
	Suppose that $h \asymp N^{-\alpha}$ for some $\alpha \in \big(0,\frac{k-2}{k} \big)$ and that $B_N=O(1)$. Then,
	\begin{equation}
		\sup_{t \in [0,1]} \bigg| N^{-1} \sum_{j=1}^{N} \xi_{N,j}(t) \bigg| = O_P\Big( N^{-1/2} h^{-1/2} \sqrt{\log N} \Big)
		\label{generic-unif-conv}
	\end{equation}
	where $\xi_{N,j}(t) = K_h(T_{j} - t) \eta_{j}(t) - E \big( K_h(T_j - t) \eta_j(t) \big)$.
\end{lem}

\begin{proof}
	For $0 < c < \frac{k-2 - k\alpha}{2k}$, let $\tilde\eta_{j}(t) = \eta_{j}(t) \mathbb{I}\big(\|\eta_{j}\|_\infty \leq N^{1/2-c} h^{1/2}\big)$ be the truncation of $\eta_{j}(t)$ by the magnitude of $N^{1/2-c} h^{1/2}$. We claim that
	\begin{equation}
	\begin{aligned}
		N^{-1} \sum_{j=1}^{N} \xi_{N,j}(t)
		&= N^{-1} \sum_{j=1}^{N} \tilde\xi_{N,j}(t) + o_P\big( N^{-1/2} h^{-1/2} \big)
	\end{aligned} \label{lem-truncation}
	\end{equation}
	uniformly for $t \in [0,1]$, where $\tilde\xi_{N,j}(t) = K_h(T_{j} - t) \tilde\eta_{j}(t) - E \big( K_h(T_j - t) \tilde\eta_j(t) \big)$. Then, it can be verified that
	\begin{equation}
		\sup_{t \in [0,1]} \bigg| N^{-1} \sum_{j=1}^{N} \tilde\xi_{N,j}(t) \bigg| = O_P \Big( N^{-1/2} h^{-1/2}\sqrt{\log N} \Big)
		\label{lem-trunc-unif-conv}
	\end{equation}
	as \eqref{generic-unif-conv}. To see this, let $\mathcal{T}_\delta(m)$ denote a finite $\delta$-covering of $[0,1]$ such that $1/\delta \leq |\mathcal{T}_\delta(m)| \leq N^m$, i.e., any $t \in [0,1]$, there exists $t' \in \mathcal{T}_\delta(m)$ such that $|t - t'| \leq N^{-m} \leq \delta$. It follows that
	\begin{equation}
	\begin{aligned}
		\sup_{t \in [0,1]} \bigg| N^{-1} \sum_{j=1}^{N} \tilde\xi_{N,j}(t) \bigg|
		&\leq \sup_{t \in \mathcal{T}_\delta(m)} \bigg| N^{-1} \sum_{j=1}^{N} \tilde\xi_{N,j}(t) \bigg| \\
		&\qquad + \sup_{t,t' \in [0,1]: \, |t-t'| \leq N^{-m}} \bigg| N^{-1} \sum_{j=1}^{N}\big( \tilde\xi_{N,j}(t) - \tilde\xi_{N,j}(t') \big) \bigg|.
 	\end{aligned} \label{lem-finite-cover1}
	\end{equation}
	We note that the second term is negligible as
	\begin{equation}
	\begin{aligned}
		&\sup_{t,t' \in [0,1]: \, |t-t'| \leq N^{-m}} \bigg| N^{-1} \sum_{j=1}^{N}\big( \tilde\xi_{N,j}(t) - \tilde\xi_{N,j}(t') \big) \bigg|\\
		&\qquad\qquad \leq 2 N^{-m} \bigg( \mathrm{Lip}(K) N^{1/2-c}h^{-3/2} + \| K \|_\infty h^{-1}  \max_{1 \leq j \leq N} \mathrm{Lip}(\eta_j) \bigg) \\
		&\qquad\qquad = O_P\Big(N^{-1/2}h^{-1/2} N^{-m} \big( N^{1 + \alpha - c} \vee N^{(1+\alpha)/2}\big) \Big)\\
		&\qquad\qquad = o_P\big( N^{-1/2}h^{-1/2} \big)
 	\end{aligned} \label{lem-finite-cover2}
	\end{equation}
	for some $m > 0$. Also, applying the standard techniques for the exponential bound of large deviations, we get
	\begin{equation}
	\begin{aligned}
		&P\bigg( \sup_{t \in \mathcal{T}_\delta(m)} \bigg| N^{-1} \sum_{j=1}^{N} \tilde\xi_{N,j}(t)\bigg|  >  C \cdot N^{-1/2} h^{-1/2}\sqrt{\log N} \bigg) \\
		&\qquad\qquad  \leq \sum_{t \in \mathcal{T}_\delta(m)}  P\bigg(  \bigg|N^{-1/2+c}h^{1/2}  \sum_{j=1}^{N} \tilde\xi_{N,j}(t)\bigg| > C \cdot N^c \sqrt{\log N} \bigg)\\
		&\qquad\qquad \leq 2N^{m + c_0 - C} \to 0 \quad (N \to \infty)
 	\end{aligned} \label{lem-lerge-dev-bound}
	\end{equation}
	for some large $C > 0$, where $c_0 = c_0(K,\alpha, c) > 0$ is a constant that depends on $K$, $\alpha$, $c$ but $\mathcal{T}_\delta(m)$. Therefore, \eqref{lem-finite-cover1} together with \eqref{lem-finite-cover2} and \eqref{lem-lerge-dev-bound} gives \eqref{lem-trunc-unif-conv}.
	
	Now, we prove the claim \eqref{lem-truncation}. Define $\mathcal{E}_j = \big( \|\eta_j\|_\infty \leq N^{1/2-c} h^{1/2} \big)$ for $j=1, \ldots, N$. It follows from Markov's inequality that
	\begin{equation}
	\begin{aligned}
		P\bigg( \bigcap_{j=1}^N \mathcal{E}_j \bigg)
		&\geq 1 - \sum_{j=1}^N P\big( \|\eta_j\|_\infty >N^{1/2-c} h^{1/2} \big)\\
		&\geq 1 - B_N N^{1 - k\big(\frac{1}{2}- c\big)} h^{-k/2} \\
		&= 1 - B_N N^{-k \big( \frac{k-2-k\alpha}{2k} - c \big)} 
		\to 1 \quad (N \to \infty).
	\end{aligned} \label{lem-unif-prob-bound}
	\end{equation}
	This implies that the $\eta_j(t)$ and $\tilde\eta_j(t)$ are equivalent to each other with probability tending to $1$ uniformly for $t \in [0,1]$. We also note that
	\begin{equation}
	\begin{aligned}
		&\sup_{t \in [0,1]} \Big| E \Big(K_h(T_j - t) \big(\eta_j(t) - \tilde\eta_j(t) \big)\Big) \Big| \\
		&\qquad\qquad = \sup_{t \in [0,1]} \Big| E \Big(K_h(T_j - t) \eta_j(t) \mathbb{I}\big( \|\eta_j\|_\infty > N^{1/2-c} h^{1/2} \big)\Big) \Big|\\
		&\qquad\qquad \leq \| K \|_\infty B_N N^{-(k-1)\big(\frac{1}{2} - c\big)}h^{-1-(k-1)/2} \\
		&\qquad\qquad \leq \| K \|_\infty B_N N^{-c -k\big( \frac{k -2 - k\alpha}{2k} - c\big)} N^{-1/2} h^{-1/2}   
		 = o\big( N^{-1/2} h^{-1/2} \big).
 	\end{aligned} \label{lem-unif-expectation}
	\end{equation}
	Finally, \eqref{lem-unif-prob-bound} and \eqref{lem-unif-expectation} imply \eqref{lem-truncation}, which completes the proof of the lemma.
\end{proof}

\begin{lem} \label{thm-kernel-unif-conv}
	Let $\mu(t) = E Y(t)$ be continuously twice differentiable in $t \in [0,1]$, where $\| \mu' \|_\infty$ and $\| \mu'' \|_\infty$ exist and are finite. Suppose that $E \|Y\|_\infty^k < \infty$ for some $k > 2$ and that $ \max_{1 \leq i \leq n} \| Y' \|_\infty$ is bounded in probability. If $P(N < a_n) = o(n^{-1})$, where $a_n \asymp n^\theta$ for some $\theta > 0$, then
	\begin{equation}
		\tilde{Y}_i^\ast(t) - Y_i(t) = O_P\Big( r_n(t) + n^{-\theta/2} h^{-1/2} \sqrt{\log n} \Big) \quad (i=1, \ldots, n)
		\label{thm-unif-rate-eq}
	\end{equation}
	uniformly for $t \in [0,1]$ , where $h \asymp n^{-\theta\alpha}$ for some $\alpha \in \big(0,\frac{k-2}{k}\big)$, $r_n(t) \asymp h^2$ if $t \in [h,1-h]$, and $r_n(t) \asymp h$ otherwise.
\end{lem}

\subsubsection*{Proof of Lemma \ref{thm-kernel-unif-conv}}
	Let $\{ S_n: n \geq 1 \}$ be a sequence of events defined as $S_n = (N_i \geq a_n \,\, \textrm{for all} \,\, i=1, \ldots, n)$. Then, $P(S_n) \geq 1 - \sum_{i=1}^n P(N_i < a_n ) \to 1$ as $n \to \infty$. We claim that the stochastic expansion of \eqref{thm-unif-rate-eq} holds conditioning on $S_n$. Then, the theorem follows since $N_i$'s are independent of $(\mathbf{Y}_i^\ast, \mathbf{T}_i, \mathbf{X}_i, \mathbf{Z}_i)$'s, where $P(S_n) \to 1$ as $n \to \infty$. For simplicity, we may assume that $N_1, \ldots, N_n$ are deterministic integers bounded below from $n^{\delta}$.
	
	To prove the stochastic expansion of \eqref{thm-unif-rate-eq}, let $\hat\lambda_i(t) = N_i^{-1} \sum_{j=1}^{N_i} K_h(T_{i,j} - t)$ denote the kernel density estimator of $\lambda(t)$ and define
	\begin{equation}
	\begin{aligned}
		\hat{g}_i^A(t)
		&= \hat\lambda_i(t)^{-1} N_i^{-1} \sum_{j=1}^{N_i} K_h(T_{i,j} - t)\varepsilon_{i,j},\\
		\hat{g}_i^B(t)
		&= \hat\lambda_i(t)^{-1} N_i^{-1} \sum_{j=1}^{N_i} K_h(T_{i,j} - t)\big( Y_i(T_{i,j}) - Y_i(t) \big),
	\end{aligned}
	\end{equation}
	where $\varepsilon_{i,j} = Y_{i,j}^\ast - Y_i(T_{i,j})$ in \eqref{model-longitudinal}, so that we re-write $\tilde{Y}_i^\ast(t) - Y_i(t)  = \hat{g}_i^A(t) + \hat{g}_i^B(t)$. It follows from Lemma \ref{lem-unif-conv} that
	\begin{equation}
	\begin{aligned}
		&N_i^{-1} \sum_{j=1}^{N_i} K_h(T_{i,j} - t) = N_i^{-1} \sum_{j=1}^{N_i} E (K_h(T_{i,j} - t) \big) + O_P\Big( n^{-\theta/2} h^{-1/2} \sqrt{\log n} \Big),\\
		&N_i^{-1} \sum_{j=1}^{N_i} K_h(T_{i,j} - t)\varepsilon_{i,j}
		= O_P\Big( n^{-\theta/2} h^{-1/2} \sqrt{\log n} \Big),\\
		&N_i^{-1} \sum_{j=1}^{N_i} K_h(T_{i,j} - t) \eta_{i,j}(t)\\
		&\qquad= N_i^{-1} \sum_{j=1}^{N_i} E \big( K_h(T_{i,j} - t) \eta_{i,j}(t) \big) + O_P\Big( n^{-\theta/2} h^{1/2} \sqrt{\log n} \Big)
	\end{aligned} \label{prop-gA-gB}
	\end{equation}
	uniformly for $t \in [0,1]$, where $\eta_{i,j}(t) = Y_i(T_{i,j}) - Y_i(t)$. We note that the magnitude of stochastic remainders are of the same order for all $i=1, \ldots, n$ as $(\mathbf{Y}_i^\ast, \mathbf{T}_i, \mathbf{X}_i, \mathbf{Z}_i)$ are iid. 
	By the standard theory of kernel smoothing, we also get
	\begin{equation}
	\begin{aligned}
	&E (K_h(T_{i,j} - t) \big)
		= \kappa_0(t) \lambda(t) + o(1), \\
	&E \big( \eta_{i,j}(t) \big)
		=
	\left\{
	\begin{array}{ll}
		h^2 \big\{ \frac{1}{2} \mu''(t) \lambda(t) + \mu'(t) \lambda'(t) \big\} \kappa_2(t) + o(h^2)	&	\textrm{if} \,\,\, t \in [h,1-h],\\
		h \mu'(t) \lambda(t) \kappa_1(t) + o(h)	&	\textrm{otherwise},
	\end{array}
	\right.
	\end{aligned} \label{eq-prop-mean}
	\end{equation}
	uniformly for $t \in [0,1]$, where
	\begin{equation}
	\kappa_r(t)
	=
	\left\{
	\begin{array}{ll}
		\int_{-\frac{t}{h}}^1 u^r K(u) \, \mathrm{d}u	&	\textrm{if} \,\,\, t \in [0,h),\\
		\int_{-1}^1 u^r K(u) \, \mathrm{d}u	&	\textrm{if} \,\,\, t \in [h,1-h],\\
		\int_{-1}^{\frac{1-t}{h}} u^r K(u) \, \mathrm{d}u	&	\textrm{if} \,\,\, t \in (1-h,1]
	\end{array}
	\right.
	\end{equation}
	for $r=0,1,2$. 
	Since $\kappa_0(t)$ does not vanish and $|\kappa_1(t)|$ and $|\kappa_2(t)|$ are bounded, we get \eqref{thm-unif-rate-eq}.

\subsubsection*{Proof of Theorem \ref{thm-kernel-test-consistent}}
    Recall that
    \begin{equation}
    \begin{aligned}
        T_n^\ast
        &= \int_0^1 \big((\boldsymbol{\mathcal{I}} - \boldsymbol{\mathcal{L}}) \tilde{\boldsymbol\beta}^\ast\big)(t)^\top \big(\tilde{\mathbb{X}}^\top \tilde{\mathbb{X}} \big) \big((\boldsymbol{\mathcal{I}} - \boldsymbol{\mathcal{L}}) \tilde{\boldsymbol\beta}^\ast\big)(t) \, \mathrm{d}t\\
        &= \int_0^1 \Big\| \big(\tilde{\mathbb{X}}^\top \tilde{\mathbb{X}} \big)^{1/2} \big((\boldsymbol{\mathcal{I}} - \boldsymbol{\mathcal{L}})\tilde{\boldsymbol\beta}^\ast\big)(t) \Big\|^2 \, \mathrm{d}t\\
        &= \int_0^1 \sum_{k=1}^p  \Big( \mathbf{e}_k^\top(\tilde{\mathbb{X}}^\top \tilde{\mathbb{X}})^{-1} \tilde{\mathbb{X}}^\top \tilde{\mathbf{Y}}^\ast(t) \Big)^2 \, \mathrm{d}t,
    \end{aligned}
    \end{equation}
    where $\mathbf{e}_k \in \mathbb{R}^p$ be a unit vector whose $k$-th component is $1$.
    We note that $\tilde\beta_j^\ast(t) =  \hat\beta_j(t) + \mathbf{e}_j^\top(\tilde{\mathbb{X}}^\top \tilde{\mathbb{X}})^{-1} \tilde{\mathbb{X}}^\top \big( \tilde{\mathbf{Y}}^\ast(t) - \mathbf{Y}(t)\big)$ for each $j=1, \ldots, p$, where $\hat\beta_j(t) = \mathbf{e}_j^\top(\tilde{\mathbb{X}}^\top \tilde{\mathbb{X}})^{-1} \tilde{\mathbb{X}}^\top\mathbf{Y}(t)$ is the least-squares estimator of $\beta_j(t)$ with fully observed data. The large sample property of $\hat{\boldsymbol\beta}(t) = (\hat\beta_1(t), \ldots, \hat\beta_p(t))^\top$ follows from Theorem \ref{cor:regression} by letting $\mathscr{I}_i = [0,1]$ for all $i=1, \ldots, n$. 
    Since $\| \mathcal{I} - \mathcal{L}\|_{\mathrm{op}} \leq 1$, it follows from  the Cauchy-Schwarz inequality and Lemma \ref{thm-kernel-unif-conv}  that 
	\begin{equation}
	\begin{aligned}
		&\int_0^1 \Big\| \mathbf{e}_j^\top \big(\tilde{\mathbb{X}}^\top \tilde{\mathbb{X}} \big)^{1/2} \big((\boldsymbol{\mathcal{I}} - \boldsymbol{\mathcal{L}})( \tilde{\boldsymbol\beta}^\ast - \hat{\boldsymbol\beta})\big)(t) \Big\|^2 \, \mathrm{d}t\\
		&\quad \leq \,\, \mathbf{e}_j^\top \tilde{\mathbb{X}}^\top \tilde{\mathbb{X}} \mathbf{e}_j \sum_{k=1}^p \int_0^1 \Big( \mathbf{e}_k^\top(\tilde{\mathbb{X}}^\top \tilde{\mathbb{X}})^{-1} \tilde{\mathbb{X}}^\top \big( \tilde{\mathbf{Y}}^\ast(t) - \mathbf{Y}(t) \big)\Big)^2  \, \mathrm{d}t \\
		&\quad \leq \,\, \bigg\{ \mathbf{e}_j^\top \tilde{\mathbb{X}}^\top \tilde{\mathbb{X}} \mathbf{e}_j \sum_{k=1}^p \big(\mathbf{e}_k^\top(\tilde{\mathbb{X}}^\top \tilde{\mathbb{X}})^{-1}\mathbf{e}_k \big) \bigg\} \sum_{i=1}^n  \int_0^1  \big( \tilde{Y}_i^\ast(t) - Y_i(t) \big)^2 \, \mathrm{d}t .\\
		&\quad = \,\, \bigg\{ \mathbf{e}_j^\top \Psi \mathbf{e}_j \mathrm{tr}\big(\Psi^{-1}\big) + o_P(1) \bigg\} \cdot O_P\big( nh^3 + n^{1-\theta} h^{-1} \log n \big) \\
		&\quad = \,\, O_P\big( n^{1-\theta(3/5)} \big) \quad (\forall j=1, \ldots, p).
	\end{aligned} \label{thm-tn-decomp1}
	\end{equation}
 The above result is analogous to the proof of theorems in \cite{zhang2007statistical}. 
	On the other hand, \eqref{agp-beta0} gives
	\begin{equation}
	\begin{aligned}
		\big\| \mathbf{e}_j^\top\big(\tilde{\mathbb{X}}^\top \tilde{\mathbb{X}} \big)^{1/2} (\boldsymbol{\mathcal{I}} - \boldsymbol{\mathcal{L}})\hat{\boldsymbol\beta} \big\|
		= \big\| \mathbf{e}_j^\top\Psi^{1/2} \sqrt{n} \big( \boldsymbol{\mathcal{I}} - \boldsymbol{\mathcal{L}})(\hat{\boldsymbol\beta} - \boldsymbol\beta_0 \big) \big\| + o_P(1)
		= O_P(1).
	\end{aligned}\label{thm-tn-decomp2}
	\end{equation}
	for all $j=1, \ldots, p$.
	Combining \eqref{thm-tn-decomp1} and \eqref{thm-tn-decomp2}, we get
	\begin{equation}
	\begin{aligned}
		T_n^\ast
		&= \int_0^1 \Big\| \big(\tilde{\mathbb{X}}^\top \tilde{\mathbb{X}} \big)^{1/2} \big((\boldsymbol{\mathcal{I}} - \boldsymbol{\mathcal{L}}) \tilde{\boldsymbol\beta}^\ast \big) (t) \Big\|^2 \, \mathrm{d}t \\
		&= \, T_n + \int_0^1 \Big\| \big(\tilde{\mathbb{X}}^\top \tilde{\mathbb{X}} \big)^{1/2} \big((\boldsymbol{\mathcal{I}} - \boldsymbol{\mathcal{L}})( \tilde{\boldsymbol\beta}^\ast - \hat{\boldsymbol\beta})\big)(t) \Big\|^2 \, \mathrm{d}t  \\
		&\qquad + \, 2  \int_0^1 \Big[\big(\tilde{\mathbb{X}}^\top \tilde{\mathbb{X}} \big)^{1/2} \big((\boldsymbol{\mathcal{I}} - \boldsymbol{\mathcal{L}})( \tilde{\boldsymbol\beta}^\ast - \hat{\boldsymbol\beta})\big)(t)\Big]^\top \Big[ \big(\tilde{\mathbb{X}}^\top \tilde{\mathbb{X}} \big)^{1/2} \big((\boldsymbol{\mathcal{I}} - \boldsymbol{\mathcal{L}}) \hat{\boldsymbol\beta}\big)(t) \Big]\, \mathrm{d}t\\
		&= T_n  + O_P\big( n^{1-\theta(3/5)} \big).
	\end{aligned}
	\end{equation}
	This completes the proof. 

\bibliographystyle{unsrt}
\bibliography{ref}

\end{document}